\documentclass{CSML}
\pdfoutput=1
\usepackage{hyperref}
\hypersetup{hidelinks}

\usepackage{lastpage}
\newcommand{\dOi}{13(3:27)2017}
\lmcsheading{}{1--\pageref{LastPage}}{}{}%
{Nov.~30,~2016}{Sep.~19, 2017}{}

\usepackage{amsthm,amsmath,mathtools,amssymb,stmaryrd,url,eufrak}
\usepackage{xspace,accents}
\usepackage{tabularx}
\usepackage{mathpartir}

\usepackage{float}
\floatstyle{boxed} 
\restylefloat{figure}

\usepackage{multiprompt}

\theoremstyle{plain}

\theoremstyle{definition}
\newtheorem{ex}[thm]{Example}

\toggletrue{paper}

\begin{document}
\title[Bisimulations for Control Operators with Dynamic Prompt
Generation]{Environmental Bisimulations for Delimited-Control Operators with
  Dynamic Prompt Generation}

\author[A. Aristiz\'abal]{Andr\'es Aristiz\'abal\rsuper a}
\address{{\lsuper{a}}Universidad Icesi, Cali, Colombia}
\email{aaaristizabal@icesi.edu.co}

\author[D. Biernacki]{Dariusz Biernacki\rsuper b}
\address{{\lsuper{b,d}}University of Wroc\l{}aw,
  Wroc\l{}aw, Poland}
\email{\{dabi,ppolesiuk\}@cs.uni.wroc.pl}
\thanks{This work
    was partially supported by PHC Polonium and by National Science Centre,
    Poland, grant no. 2014/15/B/ST6/00619.}

\author[S. Lenglet]{Sergue\"i Lenglet\rsuper c}
\address{{\lsuper{c}}Universit\'e de Lorraine,
  Nancy, France}
\email{serguei.lenglet@univ-lorraine.fr}

\author[P. Polesiuk]{Piotr Polesiuk\rsuper d}


\begin{abstract}
  We present sound and complete environmental bisimilarities for a
  variant of Dybvig~et~al.'s calculus of multi-prompted
  delimited-control operators with dynamic prompt generation.  The
  reasoning principles that we obtain generalize and advance the
  existing techniques for establishing program equivalence in calculi
  with single-prompted delimited control.

  The basic theory that we develop is presented using Madiot~et~al.'s
  framework that allows for smooth integration and composition of
  up-to techniques facilitating bisimulation proofs. We also
  generalize the framework in order to express environmental
  bisimulations that support equivalence proofs of evaluation contexts
  representing continuations. This change leads to a novel and
  powerful up-to technique enhancing bisimulation proofs in the
  presence of control operators.
\end{abstract}

\maketitle

\section{Introduction}

Control operators for delimited continuations, introduced
independently by Felleisen~\cite{Felleisen:POPL88} and by Danvy and
Filinski~\cite{Danvy-Filinski:LFP90}, allow the programmer to delimit
the current context of computation and to abstract such a delimited
context as a first-class value.  It has been shown that all
computational effects are expressible in terms of delimited
continuations~\cite{Filinski:POPL94}, and so there exists a large body
of work devoted to this canonical control structure, including our
work on a theory of program equivalence for the operators \textshift
and
\textreset~\cite{Biernacki-Lenglet:FOSSACS12,Biernacki-Lenglet:FLOPS12,
  Biernacki-Lenglet:APLAS13,Biernacki-al:HAL15}.

In their paper on type-directed partial evaluation for typed
$\lambda$-calculus with sums, Balat et al.~\cite{Balat-al:POPL04} have
demonstrated that Gunter et al.'s delimited-control operators
$\mathsf{set}$ and $\mathsf{cupto}$~\cite{Gunter-al:FPCA95}, that
support multiple prompts along with dynamic prompt generation, can
have a practical advantage over single-prompted operators such as
\textshift and \textreset. Delimited-control operators with
dynamically-generated prompts are now available in several production
programming languages such as OCaml~\cite{Kiselyov:FLOPS10} and
Racket~\cite{Flatt-al:ICFP07}, and they have been given formal
semantic treatment in the literature. In particular, Dybvig et
al.~\cite{Dybvig-al:JFP06} have proposed a calculus that extends the
call-by-value $\lambda$-calculus with several primitives that allow
for: fresh-prompt generation, delimiting computations with a prompt,
abstracting control up to the corresponding prompt, and throwing to
captured continuations. Dybvig et al.'s building blocks were shown to
be able to naturally express most of other existing control operators
and as such they form a general framework for studying delimited
continuations. Reasoning about program equivalence in Dybvig et al.'s
calculus is considerably more challenging than in single-prompted
calculi: one needs to reconcile control effects with the intricacies
introduced by fresh-prompt generation and local visibility of prompts.

In this article we investigate the behavioral theory of a slightly
modified version of Dybvig et al.'s calculus that we call the
$\lambdabla$-calculus. One of the most natural notions of program
equivalence in languages based on the $\lambda$-calculus is
\emph{contextual equivalence}: two terms are contextually equivalent
if we cannot distinguish them when evaluated within any context. The
quantification over contexts makes this relation hard to use in
practice, so it is common to characterize it using simpler relations,
like coinductively defined \emph{bisimilarities}. As pointed out
in~\cite{Koutavas-al:MFPS11}, among the existing notions of
bisimilarities, \emph{environmental
  bisimilarity}~\cite{Sangiorgi-al:TOPLAS11} is the most appropriate
candidate to characterize contextual equivalence in a calculus with
generated resources, such as prompts in~$\lambdabla$. Indeed, this
bisimilarity features an environment which accumulates knowledge about
the terms we compare. This is crucial in our case to remember the
relationships between the prompts generated by the compared
programs. We therefore define environmental bisimilarities for
$\lambdabla$, as well as \emph{up-to techniques}, which simplify the
equivalence proof of two given programs. We do so using the recently
developed framework of Madiot et
al.~\cite{Madiot-al:CONCUR14,Madiot:PhD}, where it is simpler to prove
that a bisimilarity and its up-to techniques are \emph{sound} (i.e.,
imply contextual equivalence).

After presenting the syntax, semantics, and contextual equivalence of
the calculus in Section~\ref{s:calculus}, in Section~\ref{s:standard}
we define a sound and complete environmental bisimilarity and its
corresponding up-to techniques. In particular, we define a
bisimulation \emph{up to context}, which allows to forget about a
common context when comparing two terms in a bisimulation proof. The
bisimilarity we define is useful enough to prove, e.g., the folklore
theorem about delimited control~\cite{Biernacki-Danvy:JFP06}
expressing that the static delimited-control operators \textshift and
\textreset~\cite{Danvy-Filinski:LFP90} can be simulated by the dynamic
control operators $\mathsf{control}$ and
$\mathsf{prompt}$~\cite{Felleisen:POPL88}. The technique, however, in
general requires a cumbersome analysis of terms of the form
$\inctx{E}{e}$, where $E$ is a captured evaluation context and $e$ is
any expression (not necessarily a value). We therefore define in
Section~\ref{s:better} a refined bisimilarity, called
$\star$-bisimilarity, and a more expressive bisimulation up to
context, which allows to factor out a context built with captured
continuations. Proving the soundness of these two relations requires
us to extend Madiot et al.'s framework. We show how these new
techniques can be applied to \textshift and \textreset in
Section~\ref{s:shift-reset}, improving over the existing results for
these
operators~\cite{Biernacki-Lenglet:APLAS13,Biernacki-al:HAL15}. Finally,
we discuss related work and conclude in Section~\ref{s:conclusion}.

\iftoggle{paper}{
This article is an extended version of~\cite{Aristizabal-al:FSCD16}.
Compared to that paper, in Section~\ref{ss:LTS} we discuss in more
detail the intricacies of our treatment of public/private prompts in
the definition of the LTS, whereas Section~\ref{s:shift-reset} is
new. An accompanying research report~\cite{Aristizabal-al:Inria16}
contains the omitted proofs from Sections~\ref{s:standard},
\ref{s:better}, and~\ref{s:shift-reset}.}

\section{The Calculus $\lambdabla$}
\label{s:calculus}

The calculus we consider, called $\lambdabla$, extends the
call-by-value $\lambda$-calculus with four building blocks for
constructing delimited-control operators as first proposed by Dybvig
et al.~\cite{Dybvig-al:JFP06}.~\footnote{Dybvig et al.'s control
  operators slightly differ from their counterparts considered in this
  work, but they can be straightforwardly macro-expressed in the
  $\lambdabla$-calculus.}

\subsubsection*{Syntax.} We assume we have a countably infinite set of term
variables, ranged over by $x$, $y$, $z$, and~$k$, as well as a
countably infinite set of prompts, ranged over by $p$, $q$. Given an
entity denoted by a meta-variable $m$, we write $\seq m$ for a
(possibly empty) sequence of such entities. Expressions ($e$), values
($v$), and evaluation contexts ($E$) are defined as follows:

\vspace{3mm}
\noindent
\begin{tabularx}{\textwidth}{@{\hspace{\parindent}}rlXr@{}}
$e$ & $\bnfdef$ & $v \bnfor \app{e}{e} \bnfor \prFresh{x}{e} \bnfor 
  \prReset{v}{e} \bnfor \prWithSC{v}{x}{e} \bnfor \prPushSC{v}{e}$ 
  & (expressions) \\[1mm]
$v$ & $\bnfdef$ & $\var{x} \bnfor \lam{x}{e} \bnfor \prConst{p} \bnfor 
  \reifCtx{E}$
  & (values) \\[1mm]
$E$ & $\bnfdef$ & $\mtectx \bnfor \argectx{E}{e} \bnfor \valectx{v}{E}
  \bnfor \prRstectx{\prConst{p}}{E}$
  & (evaluation contexts)
\end{tabularx}
\vspace{3mm}

Values include captured evaluation contexts $\reifCtx{E}$,
representing delimited continuations, as well as generated prompts
$\prConst{p}$. Expressions include the four building blocks for
delimited control: $\prFresh{x}{e}$ is a prompt-generating construct,
where $x$ represents a fresh prompt locally visible in $e$,
$\prReset{v}{e}$ is a control delimiter for $e$, $\prWithSC{v}{x}{e}$
is a continuation grabbing or capturing construct, and
$\prPushSC{v}{e}$ is a throw construct.

Evaluation contexts, in addition to the standard call-by-value
contexts, include delimited contexts of the form
$\prRstectx{\prConst{p}}{E}$, and they are interpreted outside-in. We
use the standard notation $\inctx{E}{e}$ ($\inctx{E}{E'}$) for
plugging a context $E$ with an expression $e$ (with a context
$E'$). Evaluation contexts are a special case of (general) contexts,
understood as a term with a hole and ranged over by $C$.

The expressions $\lam{x}{e}$, $\prFresh{x}{e}$, and $\prWithSC{v}{x}{e}$ bind
$x$; we adopt the standard conventions concerning $\alpha$-equivalence. If $x$
does not occur in $e$, we write $\lam{\osef}{e}$, $\prFresh{\osef}{e}$, and
$\prWithSC{v}{\osef}{e}$. The set of free variables of $e$ is written
$\freeVars{e}$; a term $e$ is called closed if $\freeVars{e} = \emptyset$. We
extend these notions to evaluation contexts. A variable is called \emph{fresh}
if it is free for all the entities under consideration. We write $\promptsOf{e}$
(or $\promptsOf{E}$) for the set of all prompts that occur in $e$ (or $E$
respectively).  The set $\surPr{E}$ of surrounding prompts in $E$ is the set of
all prompts guarding the hole in $E$, defined as
$\{ \prConst{p} \,|\, \exists E_1, E_2, E =
\inctx{E_1}{\prRstectx{\prConst{p}}{E_2}} \}$.

\subsubsection*{Reduction semantics.} The reduction semantics of $\lambdabla$ is
given by the following rules:

\begin{minipage}{0.55\textwidth}
\[
\begin{array}{rcll}
\app{(\lam{x}{e})}{v} & \rawred & \subst{e}{x}{v} &\\[1mm]
\prReset{\prConst{p}}{v} & \rawred & v &\\[1mm]
\prReset{\prConst{p}}{\inctx{E}{\prWithSC{\prConst{p}}{x}{e}}} & \rawred &
  \subst{e}{x}{\reifCtx{E}} & \quad \prConst{p}\notin\surPr{E}\\[1mm]
\prPushSC{\reifCtx{E}}{e} & \rawred & \inctx{E}{e} &\\[1mm]
\prFresh{x}{e} & \rawred & \subst{e}{x}{\prConst{p}}
& \quad \prConst{p}\notin\promptsOf{e}
\end{array}
\]
\end{minipage}
\begin{minipage}{0.45\textwidth}

\begin{mathpar}
  \inferrule[Compatibility]{\red{e_1}{e_2} \\ \fresh{e_2}{e_1}{E}}
            {\red{\inctx{E}{e_1}}{\inctx{E}{e_2}}}
\end{mathpar}
\end{minipage}

\vspace{1em} The first rule is the standard $\beta_v$-reduction. The
second rule signals that a computation has been completed for a given
prompt. The third rule abstracts the evaluation context up to the
dynamically nearest control delimiter matching the prompt of the grab
operator. In the fourth rule, an expression is thrown (plugged,
really) to the captured context. Note that, like in Dybvig et al.'s
calculus, the expression $e$ is not evaluated before the throw
operation takes place. In the last rule, a prompt $\prConst{p}$ is
generated under the condition that it is fresh for~$e$.

The compatibility rule needs a side condition, simply because a prompt
that is fresh for~$e$ may not be fresh for a surrounding evaluation
context. Given three entities $m_1$, $m_2$, $m_3$ for
which~$\rawPrompts$ is defined, we write $\fresh{m_1}{m_2}{m_3}$ for
the condition $(\promptsOf{m_1} \setminus
\promptsOf{m_2})\cap\promptsOf{m_3} = \varnothing$, so the side
condition states that $E$ must not mention prompts generated in the
reduction step $\red{e_1}{e_2}$. This approach differs from the
previous work on bisimulations for resource-generating
constructs~\cite{Koutavas-Wand:POPL06,Koutavas-Wand:ESOP06,Sumii:CSL09,
  Sumii-Pierce:TCS07,Sumii-Pierce:JACM07,Benton-Koutavas:MSR2008,
  Pierard-Sumii:LICS12}, where configurations of the operational
semantics contain explicit information about the resources, typically
represented by a set. We find our way of proceeding less invasive to
the semantics of the calculus.

When reasoning about reductions in the $\lambdabla$-calculus, we rely
on the notion of {\em permutation} (a bijection on prompts), ranged
over by $\sigma$, which allows to reshuffle the prompts of an
expression to avoid potential collisions: $e$ with prompts permuted by
$\sigma$ is written $\applyPrPerm{\sigma}{e}$. E.g., we can use the
first item of the following lemma before applying the compatibility
rule, to be sure that any prompt generated by $e_1 \rawred e_2$ is not
in~$\promptsOf E$.
\begin{lem}
\label{l:red-perm}\label{l:red-perm-bis}  \label{l:fresh-perm}
Let $\sigma$ be a permutation. 
\begin{itemize}
\item If $\red{e_1}{e_2}$ then
  $\red{\applyPrPerm{\sigma}{e_1}}{\applyPrPerm{\sigma}{e_2}}$.
\item For any entities $m_1$, $m_2$, $m_3$, we have
  $\fresh{m_1}{m_2}{m_3}$ iff
  $\fresh{m_1\sigma}{m_2\sigma}{m_3\sigma}$.
\end{itemize}
\end{lem}

A closed term $e$ either uniquely, up to permutation of prompts,
reduces to a term~$e'$, or it is a normal form (i.e., there is no
$e''$ such that $\red{e}{e''}$). In the latter case, we distinguish
values, control-stuck terms $\inctx{E}{\prWithSC{p}{k}{e}}$ where
$p\not\in\surPr{E}$, and the remaining expressions that we call errors
(e.g., $\inctx{E}{\app{p}{v}}$ or
$\inctx{E}{\prWithSC{\lam{x}{e}}{k}{e'}}$). We write
$\redrtc{e_1}{e_2}$ if $e_1$ reduces to $e_2$ in many (possibly 0)
steps, 
and we write $\divOrErr{e}$ when a term $e$ diverges (i.e., there exists an
infinite sequence of reductions starting with $e$) or when it reduces (in many
steps) to an error.

\begin{ex}
  Let us assume that $u$, $v$, and $w$ are values, $e$ is an
  expression, $\prConst{p}$ and $\prConst{q}$ are two different
  prompts with $\prConst{q}$ not occurring in $u$, $v$, $w$ and
  $e$. Then the following reduction sequence illustrates how fresh
  prompts are generated (1), how delimited continuations are captured
  (2 and 4), and how expressions are thrown to captured continuations
  (3):

  \[
  \begin{array}{rcl}
  \prFresh{x}
          {\prReset
            {x}
            {\app{u}
              {(\prReset{p}{\app{v}
                  {(\prWithSC{x}{k}{\app{w}
                      {(\prPushSC{k}{(\prWithSC{\prConst{p}}
                          {\_}{e})})}})}})}}}
          & \rawred & \quad\quad\quad (1)
          \\[1mm]
          \prReset
            {\prConst{q}}
            {\app{u}
              {(\prReset{p}{\app{v}
                  {(\prWithSC{q}{k}{\app{w}
                      {(\prPushSC{k}{(\prWithSC{p}{\_}{e})})}})}})}}
          & \rawred & \quad\quad\quad (2)
          \\[1mm]
          \app{w}
              {(\prPushSC
                {\reifCtx{\app{u}{(\prReset{p}{\app{v}{\hole}})}}}
                {(\prWithSC{p}{\_}{e})})}
          & \rawred & \quad\quad\quad (3)
          \\[1mm]
          \app{w}
              {(\app{u}{(\prReset{p}{\app{v}{(\prWithSC{p}{\_}{e})}})})}
          & \rawred & \quad\quad\quad (4)
          \\[1mm]
          \app{w}
              {(\app{u}{e})}
          & &
  \end{array}
  \]

\noindent
If we throw $\prWithSC{x}{\_}{e}$ to $k$ in the initial term instead of
$\prWithSC{p}{\_}{e}$, then the reduction sequence would terminate with a
control-stuck term
$\app{w}{(\app{u}{(\prReset{p}{\app{v}{(\prWithSC{q}{\_}{e})}})})}$ that could
not be unstuck by any evaluation context. Indeed, if we plug the initial
expression modified as suggested above, e.g., in the context
$\prReset{q}{\hole}$, the compatibility rule requires that in step (1) the
generated prompt $q$ be renamed into some other prompt $r$ that does not occur
in the terms under consideration, and the corresponding reduction sequence
terminates with a control-stuck term
$\prReset{q}{\app{w}{(\app{u}{(\prReset{p}{\app{v}{(\prWithSC{r}{\_}{e})}})})}}$.
\end{ex}

When presenting more complex examples, we use the fixed-point operator
$\mathsf{fix}$, $\mathsf{let}$-construct, conditional $\mathsf{if}$
along with boolean values $\constTrue$ and $\constFalse$, and
sequencing $"\mathsf{;}"$, all defined as in the call-by-value
$\lambda$-calculus. We also use the diverging term $\Omega \defeq
\app{(\lam x {\app x x})}{(\lam x {\app x x})}$, and we define an
operator $\prEqual{}{}$ to test the equality between prompts, as
follows:
\vspace{-1mm}
\[
\begin{array}{rcl}
  \prEqual{e_1}{e_2} &
  \defeq &
  \expLet{x}{e_1}
         {\expLet{y}{e_2}
           {\prReset{x}
             {(\expSeq{(\prReset{y}{\prWithSC{x}{\osef}{\constFalse}})}{\constTrue})}}}
\end{array}  
\] 
If $e_1$ and $e_2$ evaluate to different prompts, then the grab
operator captures the context up to the outermost prompt to throw it
away, and $\constFalse$ is returned; otherwise, $\constTrue$ is
returned.

\subsubsection*{Contextual equivalence.} We now define formally what it takes for
two terms to be considered equivalent in the $\lambdabla$-calculus.
First, we characterize when two closed expressions have equivalent
observable actions in the calculus, by defining the following relation
$\eqObs$.

\begin{defi}
\label{def:obs}
We say $e_1$ and $e_2$ have equivalent observable actions, noted
$\related{\eqObs}{e_1}{e_2}$, if
\begin{enumerate}
\item $\redrtc{e_1}{v_1}$ iff $\redrtc{e_2}{v_2}$,
\item $\redrtc{e_1}{\inctx{E_1}{\prWithSC{\prConst{p_1}}{x}{e'_1}}}$
  iff $\redrtc{e_2}{\inctx{E_2}{\prWithSC{\prConst{p_2}}{x}{e'_2}}}$,
  where $p_1 \not\in \surPr{E_1}$ and  $p_2 \not\in \surPr{E_2}$,
\item $\divOrErr{e_1}$ iff $\divOrErr{e_2}$.
\end{enumerate}
\end{defi}

\noindent
We can see that errors and divergence are treated as equivalent, which
is standard.

Based on $\eqObs$, we define \emph{contextual equivalence} as follows.
\begin{defi}[Contextual equivalence]
  Two closed expressions $e_1$ and $e_2$ are contextually equivalent,
  written, $e_1 \ectxEq e_2$, if for all $E$ such that $\promptsOf{E}
  = \varnothing$, we have $\inctx{E}{e_1} \eqObs \inctx{E}{e_2}$.
\end{defi}

\noindent Contextual equivalence can be extended to open terms in a standard
way: if $\freeVars{e_1}\cup\freeVars{e_2}\subseteq\seq x$, then
$\related{\ectxEq}{e_1}{e_2}$ if
$\related{\ectxEq}{\lam{\seq x}{e_1}}{\lam{\seq x}{e_2}}$. We test terms using
only promptless contexts, because the testing context should not use prompts
that are private for the tested expressions. For example, the expressions
$\lam{f}{\app{\app{\var{f}}{\prConst{p}}}{\prConst{q}}}$ and
$\lam{f}{\app{\app{\var{f}}{\prConst{q}}}{\prConst{p}}}$ should be considered
equivalent if nothing is known from the outside about $\prConst{p}$ and
$\prConst{q}$. Prompts occur in expressions because we use reduction semantics,
also known as syntactic theory. But prompts, just as closures or continuations,
are runtime entities, and in other semantics formats such as, e.g., abstract
machines~\cite{Dybvig-al:JFP06}, they would not be a part of the syntax. 

As common in calculi with resource
generation~\cite{Sumii-Pierce:TCS07,Sumii:CSL09,Sangiorgi-al:TOPLAS11}, testing
with evaluation contexts (as in $\ectxEq$) is not the same as testing with all
contexts: we have $\prFresh{x}{\var{x}} \ectxEq \prConst{p}$, but these terms
can be distinguished by \vspace{-1mm}
\[
\expLet{f}{\lam{x}{\mtectx}}{\expIf{\prEqual{
      \app{\var{f}}{\constUnit}}{\app{\var{f}}{\constUnit}}}}{
  \expOmega}{\constUnit}\] In the rest of the article, we show how to
characterize $\ectxEq$ with environmental bisimilarities.\footnote{If
  $\ctxEq$ is the contextual equivalence testing with all contexts,
  then we can prove that $e_1 \ctxEq e_2$ iff $\lam{x}{e_1} \ectxEq
  \lam{x}{e_2}$, where $x$ is any variable. We therefore obtain a
  proof method for $\ctxEq$ as well.}

\begin{rem}
  \label{r:stuck-error}
  Definition~\ref{def:obs} distinguishes control-stuck terms from
  errors, as
  making the distinction allows comparisons with the previous work on
  \textshift and \textreset~\cite{Biernacki-al:HAL15}, where a similar
  choice is made. However, unlike in~\cite{Biernacki-al:HAL15}, the
  contextual equivalence of the present article cannot ``unstuck'' a
  control-stuck term in $\lambdabla$, as we consider promptless
  contexts, so it can be natural to treat stuck terms as errors. We
  explain how making this latter choice impacts the definitions of our
  bisimilarities in Remark~\ref{r:stuck-error-standard} and
  Remark~\ref{r:stuck-error-better}.
\end{rem}

\section{Environmental Bisimilarity}
\label{s:standard}

In this section, we propose a first characterization of $\ectxEq$
using an environmental bisimilarity.  We express the bisimilarity in
the style of~\cite{Madiot-al:CONCUR14}, using a so called first-order
labeled transition system (LTS), to factorize the soundness proofs of
the bisimilarity and its up-to techniques. We start by defining the
LTS and its corresponding bisimilarity.

\subsection{Labeled Transition System and Bisimilarity}
\label{ss:LTS}

In the original formulation of environmental
bisimulation~\cite{Sangiorgi-al:TOPLAS11}, two expressions $e_1$ and
$e_2$ are compared under some environment $\mathcal E$, which
represents the knowledge of an external observer about $e_1$ and
$e_2$. The definition of the bisimulation enforces some conditions
on~$e_1$ and $e_2$ as well as on $\mathcal E$. In Madiot et al.'s
framework~\cite{Madiot-al:CONCUR14,Madiot:PhD}, the conditions on
$e_1$, $e_2$, and $\mathcal E$ are expressed using a LTS between
\emph{states} of the form $(\env, e_1)$ (and $(\envd, e_2)$), where
$\env$ (and $\envd$) is a finite sequence of values corresponding to
the first (and second) projection of the environment $\mathcal E$.
Note that in $(\env, e_1)$, $e_1$ may be a value, and therefore a
state can be simply of the form $\env$. Transitions from states of the
form $(\env, e_1)$ (where $e_1$ is not a value) express conditions on
$e_1$, while transitions from states of the form $\env$ explain how we
compare environments.  In the rest of the paper we use $\env$, $\envd$
to range over finite sequences of values, and we write $\env_i$,
$\envd_i$ for the~$i^{\text{ th}}$ element of the sequence.  We use
$\sts$, $\stt$ to range over states.

Figure~\ref{fig:lts} presents the LTS $\lts\act$, where $\act$ ranges
over all the labels. We define $\promptsOf \env$ as $\bigcup_i
\promptsOf{\env_i}$. The transition $\ltsstuck \emhc$ uses a relation
$e \redbis e'$, defined as follows: if $\red{e}{e'}$, then $e \redbis
e'$, and if $e$ is a normal form, then $e \redbis e$.\footnote{The
  relation $\redbis$ is not exactly the reflexive closure of
  $\rawred$, since an expression which is not a normal form \emph{has}
  to reduce.} To build expressions out of sequences of values, we use
different kinds of \emph{multi-hole contexts} defined as follows.

\vspace{2mm}
\noindent\begin{tabularx}{\linewidth}{@{\hspace{\parindent}}rlXr@{}}
$\mhc$ & $\bnfdef$ & $\cval \mhc \bnfor \app \mhc \mhc \bnfor \prFresh{x}{\mhc} \bnfor 
  \prReset{\cval\mhc} \mhc \bnfor \prWithSC{\cval \mhc}{x}{\mhc} \bnfor
  \prPushSC{\cval \mhc}{\mhc}$ 
  & (contexts) \\[1mm]
$\cval \mhc$ & $\bnfdef$ & $\var{x} \bnfor \lam{x}{\mhc} \bnfor 
  \reifCtx{\emhc} \bnfor \hole_i$
  & (value contexts) \\[1mm]
$\emhc$ & $\bnfdef$ & $\mtectx \bnfor \argectx{\emhc}{\mhc} \bnfor
\valectx{\cval \mhc}{\emhc} \bnfor \prRstectx{\hole_i}{\emhc}$
  & (evaluation contexts)
\end{tabularx}
\vspace{2mm}

\noindent The holes of a multi-hole context are indexed, except for the special
hole $\hole$ of an evaluation context $\emhc$, which is in evaluation position
(that is, filling the other holes of $\emhc$ with values gives a regular
evaluation context $E$).  We write $\inctx \mhc \env$ (respectively
$\inctx {\cval \mhc} \env$ and $\inctx \emhc \env$) for the application of a
context $\mhc$ (respectively $\cval \mhc$ and $\emhc$) to a sequence $\env$ of
values, which consists in replacing $\hole_i$ with $\env_i$; we assume that this
application produces an expression (or an evaluation context in the case of
$\emhc$), i.e., each hole index in the context is smaller or equal than the size
of $\env$, and for each $\prRstectx{\hole_i}{\emhc}$ construct, $\env_i$ is a
prompt.  We write $\inctx \emhc {e, \env}$ as a shorthand for $\inctx E e$ where
$E = \inctx \emhc \env$, meaning that $e$ is put in the non-indexed hole
of~$\emhc$ (note that~$e$ may also be a value). Notice that prompts are not part
of the syntax of~$\cval \mhc$, therefore a multi-hole context does not contain
any prompt: if $\inctx \mhc \env$, $\inctx {\cval \mhc} \env$, or
$\inctx \emhc {e, \env}$ contains a prompt, then it comes from $\env$ or
$e$. Our multi-hole contexts are
promptless because~$\ectxEq$ also tests with promptless contexts.\\

\begin{figure}
\begin{mathpar}
  \inferrule{\red{e_1}{e_2} \\ \fresh{e_2}{e_1} \env}
  {(\env, e_1) \lts\tau (\env, e_2)}
  \and
  \inferrule{\env_i = \lam x e}
  {\env \ltslam i {\cval\mhc} (\env, \subst e x {\inctxe{\cval\mhc} \env})}
  \and
  \inferrule{ }
  {\env \ltsval \env}
  \and
  \inferrule{\env_i = \reifCtx E}
  {\env \ltsctx i {\mhc} (\env, \inctx E {\inctxe {\mhc} \env})}
  \and
  \inferrule{\env_i = \prConst p \\ \env_j = \prConst p}
  {\env \ltsprpt i j \env}
  \and
  \inferrule{\prConst p \notin \promptsOf \env}
  {\env \ltsnuprpt (\env, \prConst p)}
  \and
  \inferrule{\prConst p \notin \surPr {E} \\
    \inctx{\emhc}{\inctx E {\prWithSC{\prConst p} x
        e}, \env} \redbis e'}
  {(\env, \inctx E {\prWithSC{\prConst p} x
        e}) \ltsstuck{\emhc} (\env, e')}
\end{mathpar}
\caption{Labeled Transition System for $\lambdabla$}
\label{fig:lts}
\end{figure}

We now detail the rules of Figure \ref{fig:lts}, starting with the transitions
that one can find in any call-by-value
$\lambda$-calculus~\cite{Madiot-al:CONCUR14}. An internal action $(\env, e_1)
\lts\tau \sts$ corresponds to a reduction step, except we ensure that any
generated prompt is fresh w.r.t.\ $\env$. The transition $\env \ltslam i
{\cval\mhc} \sts$ signals that $\env_i$ is a $\lambda$-abstraction, which can be
tested by passing it an argument built from~$\env$ with the context
$\cval\mhc$. The transition $\ltsctx i \mhc$ for testing continuations is built
the same way, except we use a context $\mhc$, because any expression can be
thrown to a captured context. Finally, the transition $\env \ltsval \env$ means
that the state $\env$ is composed only of values; it does not test anything on
$\env$, but this transition is useful for the soundness proofs of
Section~\ref{ss:up-to}. When we have $\env \rel (\envd, e)$ (where $\rel$ is,
e.g., a bisimulation), then $(\envd, e)$ has to match with $(\envd, e)
\rtc{\lts\tau} \ltsval (\envd, v)$ so that $(\envd, v)$ is related to $\env$. We
can then continue the proofs with two related sequences of values. Such a
transition has been suggested in~\cite[Remark 5.3.6]{Madiot:PhD} to simplify the
proofs for a non-deterministic language, like $\lambdabla$.

We now explain the rules involving prompts.
When comparing two terms generating prompts, one can produce $p$ and the other a
different $q$, so we remember in $\env$, $\envd$ that $p$ corresponds to
$q$. But an observer can compare prompts using $\rawprEqual$, so $p$ has to be
related \emph{only} to~$q$.  We check it with $\ltsprpt i j$: if $\env
\ltsprpt i j \env$, then $\envd$ has to match, meaning
that $\envd_i=\envd_j$, and doing so for all $j$ such that $\env_i = \env_j$
ensures that all copies of $\env_i$ are related only to $\envd_i$.
The transition $\ltsprpt i i$ also signals that $\env_i$ is a prompt and should
be related to a prompt.  The other transition involving prompts is $\env
\ltsnuprpt (\env, p)$, which encodes the possibility for an observer to
generate fresh prompts to compare terms.  If $\env$ is related to $\envd$,
then~$\envd$ has to match by generating a prompt $q$, and we remember that $p$
is related to~$q$. For this rule to be automatically verified, we define the
\emph{prompt checking} rule for a relation $\rel$ as follows:
\vspace{1mm}
\begin{mathpar}
  \inferrule{\env \rel \envd \\ p \notin \promptsOf \env \\
    q \notin \promptsOf \envd} {(\env, p) \rel (\envd, q) }
    \;(\promptCheckRule)
\end{mathpar}
\vspace{1mm}

\noindent
Henceforth, when we construct a bisimulation $\rel$ by giving a set of
rules, we always include the $(\promptCheckRule)$ rule so that the
$\ltsnuprpt$ transition is always verified.

Finally, the transition $\ltsstuck {\emhc}$ deals with stuck terms. An
expression $\inctx E {\prWithSC p x e}$ is able to reduce if the surrounding
context is able to provide a delimiter $\rawprReset p$. However, it is possible
only if $p$ is available for the outside, and therefore is in $\env$. If
$p \notin \surPr{\inctx \emhc \env}$, then
$\inctx{\emhc}{\inctx E {\prWithSC{\prConst p} x e}, \env}$ remains stuck, and
we have
$\inctx{\emhc}{\inctx E {\prWithSC{\prConst p} x e}, \env} \redbis
\inctx{\emhc}{\inctx E {\prWithSC{\prConst p} x e}, \env}$.
Otherwise, it can reduce and we have
$\inctx{\emhc}{\inctx E {\prWithSC{\prConst p} x e}, \env} \redbis e'$, where
$e'$ is the result after the capture. The rule for $\ltsstuck \emhc$ may seem
demanding, as it tests stuck terms with all contexts $\emhc$, but up-to
techniques will alleviate this issue (see
Example~\ref{ex:stuck-utctx}). Besides, we believe testing all contexts is
necessary to be sound and complete w.r.t. contextual equivalence. Inspired by
the previous work on \textshift and
\textreset~\cite{Biernacki-Lenglet:APLAS13,Biernacki-al:HAL15}, one could
propose the following rule
\vspace{1mm}
\begin{mathpar}
  \inferrule{\prConst p \notin \surPr {E} \\ \env_i = p \\
    \red{\prReset p {\inctx{\emhc}{\inctx E {\prWithSC{\prConst p} x
        e}, \env}}}{e'}}
  {(\env, \inctx E {\prWithSC{\prConst p} x
        e}) \ltsstuck{\prReset{\hole_i}{\emhc}} (\env, e')}~(*)
\end{mathpar}
\vspace{1mm}

\noindent which tests stuck terms with context of the form $\prReset p
\emhc$, and only if $p$ is in $\env$. This rule alone is not sound, as
it would relate $(\emptyset, \Omega)$ and $(\emptyset, \inctx E
{\prWithSC{\prConst p} x e})$, because $\prConst p$ does not occur in
the environment. We could retrieve soundness by simply adding a rule
which tests if an expression is control-stuck, to deal with this kind
of situation. However, the rule $(*)$ is also too discriminating and
would break completeness, as we can see with the next two examples.

\begin{ex}
  Stuck terms may be equivalent, even though the prompts they use are not
  related in $\env$, $\envd$. For example, consider $(p_1, \expFixtwo x
  {\prWithSC{p_1} y x})$ and $(p_2, \prWithSC q y e)$, where $p_2 \neq q$ and
  $e$ is any expression. Because we can use $p_1$ to build testing contexts, we
  can trigger the capture for the first term. By doing so, we make it reduce to
  itself, while the second term remains stuck in any context. We can prove them
  bisimilar with the rules of Figure~\ref{fig:lts}. In contrast, $(p_2,
  \prWithSC q y e)$ cannot make a transition with rule $(*)$ (because $q \neq
  p_2$) while $(p_1, \expFixtwo x {\prWithSC{p_1} y x})$ can, so rule $(*)$
  would wrongfully distinguish these two expressions.
\end{ex}

\begin{ex}
  Assuming $p \neq q$, the expression
  $e_1 \defeq \prWithSC q \osef {\prWithSC p \osef v}$ aborts the current
  continuation up to the first enclosing delimiter $\rawprReset p$ which is
  behind a delimiter $\rawprReset q$, and then returns~$v$. The term
  $e_2 \defeq \expFixtwo x {\prWithSC p k {\expIf {\expInSp q k} v x}}$ has the
  same behavior: it decomposes the continuation piece by piece, repeatedly
  capturing $k$ up to $\rawprReset p$, until it finds $\rawprReset q$ in
  $k$. Testing if $\expInSp q k$ can be implemented in a similar way as testing
  prompt equality:
  $\expInSp q k \defeq \prFresh x {\prReset x {\prReset q {(\expSeq {(\prReset x
          {(\prPushSC k {\prWithSC q \osef {\prWithSC x \osef \constFalse}})})}
        \constTrue)}}} $.
  Again, the rule $(*)$ wrongfully distinguishes $(p, q, e_1)$ and $(p, q, e_2)$,
  because $e_1$ captures on $q$ first while $e_2$ captures on $p$.
\end{ex}

For weak transitions, we define $\ltswktau$ as $\rtc{\lts\tau}$, $\ltswk \act$
as $\ltswktau$ if $\act=\tau$ and as $\ltswktau \lts\act \ltswktau$
otherwise. We define bisimulation and bisimilarity using a more general notion
of \emph{progress}. Henceforth, we
let $\rel$, $\rels$ range over relations on states.

\begin{defi}
  \label{d:progress}
  A relation $\rel$ progresses to $\rels$, written $\rel \progress \rels$, if
  $\rel \mathop\subseteq \rels$ and $\sts \rel \stt$ implies
  \begin{itemize}
  \item if $\sts \lts\act \sts'$, then there exists $\stt'$ such that $\stt
    \ltswk\act \stt'$ and $\sts' \rels \stt'$;
  \item the converse of the above condition on $\stt$.
  \end{itemize}
  A \emph{bisimulation} is a relation $\rel$ such that $\rel \progress \rel$, and
  \emph{bisimilarity} $\bisim$ is the union of all bisimulations. 
\end{defi}

\subsection{Up-to Techniques, Soundness, and Completeness}
\label{ss:up-to}

Before defining the up-to techniques for $\lambdabla$, we briefly recall the
main definitions and results we use
from~\cite{Sangiorgi-Pous:11,Madiot-al:CONCUR14,Madiot:PhD}; see these works for
more details. We use $f$, $g$ to range over functions on relations on states. An
\emph{up-to technique} is a function $f$ such that $\rel \progress f(\rel)$
implies $\rel \mathop\subseteq \bisim$. However, this definition is difficult to
use to prove that a given $f$ is an up-to technique, so we rely on
\emph{compatibility} instead, which gives sufficient conditions for $f$ to be an
up-to technique.

We first define some auxiliary notions and notations. We write $f
\subseteq g$ if $f(\rel) \subseteq g(\rel)$ for all $\rel$. We define
$f \cup g$ argument-wise, i.e., $(f \cup g)(\rel)=f(\rel) \cup
g(\rel)$, and given a set $\setF$ of functions, we also write~$\setF$
for the function defined as $\bigcup_{f \in \setF} f$. We define
$f^\omega$ as $\bigcup_{n \in \mathbb N} f^n$. We write $\rawid$ for
the identity function on relations, and $\fid f$ for $f \mathop\cup
\rawid$. A function~$f$ is monotone if $\rel \mathop\subseteq \rels$
implies $f(\rel) \mathop\subseteq f(\rels)$.  We write $\finpower\rel$
for the set of finite subsets of $\rel$, and we say $f$ is continuous
if it can be defined by its image on these finite subsets, i.e., if
$f(\rel) \mathop\subseteq \bigcup_{\rels \in
  \finpower\rel}f(\rels)$. The up-to techniques of the present paper
are defined by inference rules with a finite number of premises, so
they are trivially continuous. Continuous functions are interesting
because of their properties:\footnote{Unlike in~\cite{Madiot:PhD}, we
  use $\fid f$ instead of $f$ in the last property of
  Lemma~\ref{l:continuity} (expressing idempotence of ${\fid
    f}^\omega$), as $\rawid$ has to be factored in somehow for the
  property to hold.}

\begin{lem}
  \label{l:continuity} 
  If $f$ and $g$ are continuous, then $f \comp g$ and $f \cup g$ are continuous.
  
  If $f$ is continuous, then $f$ is monotone, and
  $f \comp {\fid f}^\omega \subseteq {\fid f}^\omega$.
\end{lem}

\begin{defi}
  A function $f$ \emph{evolves} to $g$, written $f \evolve g$, if for all $\rel
  \progress \rels$, we have $f(\rel) \progress g(\rels)$. 
  A set $\setF$ of continuous functions is compatible if for all $f \in \setF$, $f \evolve
  \fid\setF^\omega$.
\end{defi}

\begin{lem}
  \label{l:properties-compatibility}
  Let $\setF$ be a compatible set, and $f \in \setF$;  $f$ is an up-to technique,
  and $f(\bisim) \mathop\subseteq \bisim$.
\end{lem}
\noindent Proving that $f$ is in a compatible set $\setF$ is easier than proving it is an
up-to technique, because we just have to prove that it evolves towards a
combination of functions in $\setF$. Besides, the second property of
Lemma~\ref{l:properties-compatibility} can be used to prove that~$\bisim$ is a
congruence just by showing that bisimulation up to context is compatible.\\

The first technique we define allows to forget about prompt names; in a
bisimulation relating $(\env, e_1)$ and $(\envd, e_2)$, we remember that
$\env_i=p$ is related to $\envd_i=q$ by their position~$i$, not by their
names. Consequently, we can apply different permutations to the two states to
rename the prompts without harm, and bisimulation \emph{up to
  permutations}\footnote{Madiot defines a bisimulation ``up to permutation'' in
  \cite{Madiot:PhD} which reorders values in a state. Our bisimulation up to
  permutations operates on prompts.} allows us to do so. It is reminiscent of
bisimulation up to renaming~\cite{Sumii:CSL09}, which operates on reference
names. Given a relation $\rel$, we define $\perm\rel$ as
$\sts\sigma_1 \perm\rel \stt\sigma_2$, assuming $\sts \rel \stt$ and $\sigma_1$,
$\sigma_2$ are any permutations.

We then allow to remove or add values from the states with,
respectively, bisimulation \emph{up to weakening} $\rawweak$ and
bisimulation \emph{up to strengthening} $\rawstr$, defined as follows
\vspace{1mm}
\begin{mathpar}
  \inferrule{(\seq v, \env, e_1) \rel (\seq w, \envd, e_2)} {(\env, e_1) \weak\rel
    (\envd, e_2)} 
  \and
  \inferrule{(\env, e_1) \rel (\envd, e_2)}
  {(\env, \inctx{\cval \mhc}{\env}, e_1) \str\rel (\envd, \inctx{\cval
      \mhc}{\envd}, e_2)}
\end{mathpar}
\vspace{1mm}

\noindent
Bisimulation up to weakening diminishes the testing power of states,
since less values means less arguments to build from for the
transitions $\ltslam i {\cval \mhc}$, $\ltsctx i \mhc$, and $\ltsstuck
\emhc$. This up-to technique is usual for environmental bisimulations,
and is called ``up to environment''
in~\cite{Sangiorgi-al:TOPLAS11}. In contrast, $\rawstr$ adds values to
the state, but without affecting the testing power, since the added
values are built from the ones already in $\env$, $\envd$.

Finally, we define the well-known bisimulation up to context, which
allows to factor out a common context when comparing terms. As usual
for environmental bisimulations~\cite{Sangiorgi-al:TOPLAS11}, we
define two kinds of bisimulation up to context, depending whether we
operate on values or any expressions. For values, we can factor out
any common context $\mhc$, but for expressions that are not values, we
can factor out only an evaluation context $\emhc$, since factoring out
any context in that case would lead to an unsound up-to
technique~\cite{Madiot:PhD}. We define up to context for values
$\rawutctxv$ and for any expression $\rawutctx$ as follows:
\vspace{1mm}
\begin{mathpar}
\inferrule{\env \rel \envd}
{(\env, \inctx{\mhc}{\env}) \utctxv\rel (\envd, \inctx{\mhc}{\envd})}
\and
\inferrule{(\env, e_1) \rel (\envd, e_2)}
{(\env, \inctx{\emhc}{e_1, \env}) \utctx\rel (\envd, \inctx{\emhc}{e_2, \envd})}
\end{mathpar}
\vspace{1mm}

\begin{lem}
  The set $\{\rawperm, \rawweak, \rawstr, \rawutctxv, \rawutctx \}$ is
  compatible. 
\end{lem}
\noindent The function $\rawutctx$ particularly helps in dealing with stuck
terms, as we can see below.

\begin{ex}
  \label{ex:stuck-utctx}
  Let $\sts \defeq (\env, \prWithSC p x {e_1})$ and $\stt \defeq (\envd,
  \prWithSC q x {e_2})$ (for some $e_1$, $e_2$), so that $\sts \rel \stt$. If
  $p$ and $q$ are not in $\env$, $\envd$, then the two expressions remain stuck,
  as we have $\sts \ltsstuck \emhc (\env, \inctx \emhc {\prWithSC p x {e_1},
    \env})$ and similarly for $\stt$. We have directly $(\env, \inctx \emhc
  {\prWithSC p x {e_1}, \env}) \utctx \rel (\envd, \inctx \emhc {\prWithSC q x
    {e_2}, \envd})$. Otherwise, the capture can be triggered with a context
  $\emhc$ of the form $\inctx{\emhc_1}{\prReset{\hole_i}{\emhc_2}}$, giving
  $\sts \ltsstuck \emhc (\env, \inctx {\emhc_1}{\subst {e_1} x
    {\reifCtx{\inctx{\emhc_2}{\env}}}, \env})$ and $\stt \ltsstuck \emhc (\envd,
  \inctx {\emhc_1}{\subst {e_2} x {\reifCtx{\inctx{\emhc_2}{\envd}}},
    \envd})$. Thanks to $\rawutctx$, we can forget about $\emhc_1$ which does
  not play any role, and continue the bisimulation proof by focusing only on
  $(\env, \subst {e_1} x {\reifCtx{\inctx{\emhc_2}{\env}}})$ and $(\envd, \subst
  {e_2} x {\reifCtx{\inctx{\emhc_2}{\envd}}})$.
\end{ex}

Because bisimulation up to context is compatible,
Lemma~\ref{l:properties-compatibility} ensures that $\bisim$ is a
congruence w.r.t. all contexts for values, and w.r.t. evaluation
contexts for all expressions. As a corollary, we can deduce that
$\bisim$ is sound w.r.t. $\ectxEq$; we can also prove that it is
complete w.r.t. $\ectxEq$, leading to the following full
characterization result.

\begin{thm}
  \label{t:charac-std}
  $e_1 \ectxEq e_2$ iff $(\emptyset, e_1) \bisim (\emptyset, e_2)$.
\end{thm}
\noindent For completeness, we prove that \( \{ (\env, e_1), (\envd, e_2) \bnfor
\forall \emhc, \inctx \emhc {e_1, \env} \eqObs \inctx \emhc {e_2,
  \envd} \} \) is a bisimulation up to permutation.

\begin{rem}
  \label{r:stuck-error-standard}
  If we consider control-stuck terms as errors, as suggested in
  Remark~\ref{r:stuck-error}, then a control-stuck term that cannot be
  unstuck can be related to a term that reduces to another kind of
  error or that diverges. To take this into account, we would change
  the transition $\ltsstuck \emhc$ as follows:
  \begin{mathpar}
    \inferrule{\inctx{\emhc}{e, \env} \redgrab e'}
              {(\env, e) \ltsstuck{\emhc} (\env, e')}
  \end{mathpar}
  where the relation $\redgrab$ differs from $\redbis$ in that it enforces only
  the continuation-grabbing reduction step, if possible. Note that this
  transition is useless when comparing two states $(\env, e_1)$ and
  $(\envd, e_2)$ where neither $e_1$ nor $e_2$ is stuck, but in that case, we
  obtain $(\env, \inctx \emhc {e_1, \env})$ and
  $(\envd, \inctx \emhc {e_2, \envd})$, which are directly related by $\rawutctx$.

  With such a change, the results of this section remain valid with respect to
  the notion of contextual equivalence defined in
  Remark~\ref{r:stuck-error}. The proofs are almost the same, since the extra
  cases involving the $\ltsstuck \emhc$ transition applied to expressions that
  are not control-stuck can be dealt with using $\rawutctx$, as explained
  above. 
\end{rem}

\subsection{Example}
\label{ss:example-std}

As an example, we show a folklore theorem about delimited
control~\cite{Biernacki-Danvy:JFP06}, stating that the static
operators \textshift and \textreset can be simulated by the
dynamic operators $\mathsf{control}$ and $\mathsf{prompt}$. In fact, what
we prove is a more general and stronger result than the original
theorem, since we demonstrate that this simulation still holds when
multiple prompts are around.

\begin{ex}[Folklore theorem]
  \label{ex:folklore}
  We encode \textshift, \textreset, $\mathsf{control}$, and
  $\mathsf{prompt}$ as follows
\[
\begin{array}{lcllcl}
\mathsf{shift}_p &\defeq &
    \lam{f}{\prWithSC{p}{k}{\prReset{p}{\app{\var{f}{
        (\lam{y}{\prReset{p}{\prPushSC{\var{k}}{\var{y}}}})}}}}} \hspace*{2em}&
\mathsf{control}_p &\defeq &
    \lam{f}{\prWithSC{p}{k}{\prReset{p}{\app{\var{f}{
        (\lam{y}{\prPushSC{\var{k}}{\var{y}}}})}}}} \\

\mathsf{reset}_p &\defeq & \reifCtx{\prRstectx{p}{\mtectx}} & \mathsf{prompt}_p
&\defeq & \reifCtx{\prRstectx{p}{\mtectx}} 
\end{array}
\]
Let $\mathsf{shift'}_p \defeq \lam{f}{\app{\mathsf{control}_p}{(\lam{l}{
      \app{\var{f}}{(\lam{z}{\prPushSC{\mathsf{prompt}_p}{
            \app{\var{l}}{\var{z}}}})}})}}$; we prove that \iftoggle{paper}{the pair}{}
$(\mathsf{shift}_p, \mathsf{reset}_p)$ (encoded as
$\lam{f}{\app{\app{\var{f}}{\mathsf{shift}_p}}{\mathsf{reset}_p}}$) is bisimilar
to $(\mathsf{shift'}_p, \mathsf{prompt}_p)$ (encoded as
$\lam{f}{\app{\app{\var{f}}{\mathsf{shift'}_p}}{\mathsf{prompt}_p}}$).
\end{ex}

\begin{proof}
  We iteratively build a relation $\rel$ closed under $(\promptCheckRule)$ such
  that $\rel$ is a bisimulation up to context, starting with
  $(p,\mathsf{shift}_p) \rel (p,\mathsf{shift'}_p)$. The transition
  $\ltsprpt{1}{1}$ is easy to check. For $\ltslam{2}{\cval\mhc}$, we obtain
  states of the form $(p,\mathsf{shift}_p, e_1)$, $(p,\mathsf{shift'}_p,
  e_2)$ that we add to $\rel$, where~$e_1$ and $e_2$ are the terms below
\vspace{1mm}
\begin{mathpar}
\inferrule{\env \rel \envd}
{(\env,\prWithSC{p}{k}\prReset{p}\app{\inctxe{\cval\mhc}\env}{
    (\lam{y}\prReset{p}\prPushSC{\var{k}}\var{y})}) \rel
 (\envd,\prWithSC{p}{k}\prReset{p}
 \app{(\lam{l}\app{\inctxe{\cval\mhc}\envd}
 {(\lam{z}{\prPushSC{\mathsf{prompt}_p}{
   \app{\var{l}}{\var{z}}}})})}
 {(\lam{y}{\prPushSC{\var{k}}\var{y}})})}
\end{mathpar}
\vspace{1mm}

\noindent
We use an inductive, more general rule, because we want $\ltslam{2}{\cval\mhc}$
to be still verified after we extend $(p,\mathsf{shift}_p)$ and
$(p,\mathsf{shift'}_p)$. The terms $e_1$ and $e_2$ are stuck, so we test them
with $\ltsstuck\emhc$. If $\emhc$ does not trigger the capture, we obtain
$\inctx \emhc{e_1,\env}$ and $\inctx \emhc {e_2,\envd}$, and we can use
$\rawutctx$ to conclude. Otherwise, $\emhc =
\inctx{\emhc'}{\prRstectx{\hole_1}{\emhc''}}$ (where $\prRstectx{\hole_1}{}$
does not surround $\hole$ in $\emhc''$), and we get
\[
\inctxe{\emhc'}{\prReset{p}\app{\inctxe{\cval\mhc}\env}{
    (\lam{y}\prReset{p}\prPushSC{\reifCtx{\inctxe{\emhc''}{\env}}}\var{y}), \env}}
\mbox{ and }
\inctxe{\emhc'}{\prReset{p}\app{\inctxe{\cval\mhc}\envd}{
    (\lam{z}\prPushSC{\mathsf{prompt}_p}\app{
        (\lam{y}\prPushSC{\reifCtx{\inctxe{\emhc''}{\envd}}}\var{y})}{
        \var{z}}), \envd}}
\]
We want to use $\rawutctxv$ to remove the common context
$\inctx{\emhc'}{\prRstectx{\hole_1}\valectx{\cval\mhc}{\hole_i}}$, which means that we
have to add the following states in the definition of $\rel$ (again, inductively):
\vspace{1mm}
\begin{mathpar}
\inferrule{\env\rel\envd}
    {(\env, 
      \lam{y}\prReset{p}\prPushSC{\reifCtx{\inctxe{\emhc''}{\env}}}\var{y})
    \rel
    (\envd,
      \lam{z}\prPushSC{\mathsf{prompt}_p}\app{
        (\lam{y}\prPushSC{\reifCtx{\inctxe{\emhc''}{\envd}}}\var{y})}{
        \var{z}})}
\end{mathpar}
\vspace{1mm}

\noindent
Testing these functions with $\ltslam{i}{\cval\mhc}$ gives on both sides states
where $\prReset{\hole_1}\inctxe{\emhc''}{\cval\mhc}$ can be removed with
$\rawutctxv$. Because \( (\varnothing,
\lam{f}{\app{\app{\var{f}}{\mathsf{shift}_p}}{\mathsf{reset}_p}})
\weak{\utctxv{\rel}} (\varnothing,
\lam{f}{\app{\app{\var{f}}{\mathsf{shift'}_p}}{\mathsf{prompt}_p}}) \), it is
enough to conclude. Indeed, $\rel$ is a bisimulation up to context, so $\rel
\mathop\subseteq \bisim$, which implies $\weak{\utctxv\rel} \mathop\subseteq
\weak{\utctxv\bisim}$
(because $\rawweak$ and $\rawutctxv$ are monotone), and $\weak{\utctxv{\bisim}}
\mathop\subseteq \bisim$ (by Lemma~\ref{l:properties-compatibility}). Note that this
reasoning works for any combination of monotone up-to techniques and any
bisimulation (up-to).
\end{proof}

What makes the proof of Example~\ref{ex:folklore} quite simple is that we relate
$(p, \mathsf{shift}_p)$ and $(p, \mathsf{shift'}_p)$, meaning that $p$ can be
used by an outside observer. But the control operators $(\mathsf{shift}_p,
\mathsf{reset}_p)$ and $(\mathsf{shift'}_p, \mathsf{prompt}_p)$ should be the
only terms available for the outside, since~$p$ is used only to implement them.
If we try to prove equivalent these two pairs directly, i.e., while keeping $p$
private, then 
testing $\mathsf{reset}_p$ and $\mathsf{prompt}_p$ with $\ltsctx{2}{\mhc}$
requires a cumbersome analysis of the behaviors of $\prReset p
{\inctx{\mhc}{\env}}$ and $\prReset p {\inctx{\mhc}{\envd}}$.  In the next
section, we define a new kind of bisimilarity with a powerful up-to technique to
make such proofs more tractable.

\section{The $\star$-Bisimilarity}
\label{s:better}

In this section we develop a refined version of bisimilarity along
with a powerful up to context technique for the $\lambdabla$-calculus
that relies on testing captured continuations with values only,
instead of with arbitrary expressions. In order to account for such an
enhancement we generalize Madiot's framework.

\subsection{Motivation}
\label{ss:motivation}

Let us start with identifying some drawbacks of the existing
environmental bisimulation techniques for control operators, such as
the one of Section~\ref{s:standard} and the ones
of~\cite{Biernacki-Lenglet:APLAS13,Biernacki-al:HAL15}, in the way
captured contexts are tested and exploited.

\subsubsection*{Testing continuations.} In Section~\ref{s:standard}, a
continuation $\env_i = \reifCtx E$ is tested with
$\env \ltsctx i \mhc (\env, \inctx{E}{\inctx \mhc \env})$. We then have to study
the behavior of $\inctx E {\inctx \mhc \env}$, which depends primarily on how
$\inctx \mhc \env$ reduces; e.g., if $\inctx \mhc \env$ diverges, then $E$ does
not play any role. Consequently, the transition $\ltsctx i \mhc$ does not really
test the continuation directly, since we have to reduce $\inctx \mhc \env$
first.  To really exhibit the behavior of the continuation, we change the
transition so that it uses a value context instead of a general one. We then
have $\env \ltsctx i {\cval \mhc} (\env, \inctx{E}{\inctx {\cval \mhc} \env})$,
and the behavior of the term we obtain depends primarily on $E$. However, this
is not equivalent to testing with $\mhc$, since $\inctx \mhc \env$ may interact
in other ways with $E$ if $\inctx \mhc \env$ is a stuck term. If~$E$ is of the
form $\inctx{E'}{\prReset p {E''}}$ with $p \notin \surPr{E''}$, and $p$ is in
$\env$, then $\mhc$ may capture~$E''$, since $p$ can be used to build an
expression of the form $\prWithSC p x e$. To take into account this possibility,
we introduce a new transition
$\env \ltsdecomp i j (\env, \reifCtx{E'}, \reifCtx{E''})$, which decomposes
$\env_i=\inctx{E'}{\prReset p {E''}}$ into $\reifCtx{E'}$ and~$\reifCtx{E''}$,
provided $\env_j=p$. The stuck term $\inctx \mhc \env$ may also capture $E$
entirely, as part of a bigger context of the form
$\inctx {\emhc_1}{\inctx E {\emhc_2}}$. To take this into account, we introduce
a way to build such contexts using captured continuations. This is also useful
to make bisimulation up to context more expressive, as we explain in the next
paragraph.

\subsubsection*{A more expressive bisimulation up to context.} As we already
pointed out in~\cite{Biernacki-Lenglet:APLAS13,Biernacki-al:HAL15}, bisimulation
up to context is not very helpful in the presence of control operators. For
example, suppose we prove the $\beta_\Omega$ axiom
of~\cite{Kameyama-Hasegawa:ICFP03}, i.e., $\app{(\lam x {\inctx E x})}{e}$ is
equivalent to $\inctx E e$ if $x \notin \freeVars E$ and
$\surPr{E} = \emptyset$. If~$e$ is a stuck term $\prWithSC p y {e_1}$, we have
to compare
$\subst {e_1} y {\reifCtx{\inctx {E_1}{\app{(\lam x {\inctx E x})} \hole}}}$ and
$\subst {e_1} y {\reifCtx{\inctx {E_1}{E}}}$ for some $E_1$. If
\iftoggle{paper}{$e_1$ is of the form
  $\prPushSC y {(\prPushSC y {e_2})}$}{$e_1=\prPushSC y {(\prPushSC y {e_2})}$},
then we get respectively
\[\inctx {E_1}{\app{(\lam x {\inctx E x})}{\inctx{E_1}{\app{(\lam x {\inctx E
          x})}{e_2}}}}
\text{ and }\inctx {E_1}{\inctx E {\inctx{E_1}{\inctx E {e_2}}}}.\]
We can see that the
two resulting expressions have the same shape, and yet we can only remove the
outermost occurrence of $E_1$ with $\rawutctx$. The problem is that bisimulation
up to context can factor out only a \emph{common} context. We want an up-to
technique able to identify \emph{related} contexts, i.e., contexts built out of
related continuations. To do so, we modify the multi-hole contexts to include a
construct $\appcont{\holecont_i}{\mhc}$ with a special hole $\holecont_i$, which
can be filled only with $\reifCtx E$ to produce a context $\inctx E \mhc$. As a
result, if $\env = (\reifCtx {\app{(\lam x {\inctx E x})} \hole})$ and
$\envd = (\reifCtx E)$, then
$\inctx {E_1}{\app{(\lam x {\inctx E x})}{\inctx{E_1}{\app{(\lam x {\inctx E
          x})} \hole}}}$
and $\inctx {E_1}{\inctx E {\inctx{E_1}{\inctx E \hole}}}$ can be written
$\inctx \emhc \env$, $\inctx \emhc \envd$ with
$\emhc = \inctx {E_1}{\appcont{\holecont_1}{\inctx {E_1}{\appcont{\holecont_1}
      \hole}}}$. We can then focus only on testing $\env$ and $\envd$.

However, such a bisimulation up to related context would be unsound if
not restricted in some way. Indeed, let $\reifCtx{E_1}$,
$\reifCtx{E_2}$ be any continuations, and let $\env =
(\reifCtx{E_1})$, $\envd = (\reifCtx{E_2})$. Then the transitions
$\env \ltsctx 1 {\cval \mhc}(\env, \inctx {E_1}{\inctx {\cval \mhc}
  \env})$ and $\envd \ltsctx 1 {\cval \mhc} (\envd, \inctx
{E_2}{\inctx {\cval \mhc} \envd})$ produce states of the form $(\env,
\inctx {\mhc} \env)$, $(\envd, \inctx {\mhc} \envd)$ with $\mhc =
\appcont {\holecont_1}{\cval \mhc}$. If bisimulation up to related
context was sound in that case, it would mean that $\reifCtx{E_1}$ and
$\reifCtx{E_2}$ would be bisimilar for all $E_1$ and~$E_2$, which, of
course, is wrong.\footnote{The problem is similar if we test
  continuations using contexts $\mhc$ (as in Section~\ref{s:standard})
  instead of $\cval \mhc$.} To prevent this, we distinguish
\emph{passive} transitions (such as $\ltsctx i {\cval \mhc}$) from the
other ones (called \emph{active}), so that only selected up-to
techniques (referred to as \emph{strong}) can be used after a passive
transition. In contrast, any up-to technique (including this new
bisimulation up to related context) can be used after an active
transition. To formalize this idea, we have to extend Madiot et al.'s
framework to allow such distinctions between transitions and between
up-to techniques.

\subsection{Labeled Transition System and Bisimilarity}

First, we explain how we alter the LTS of Section~\ref{ss:LTS} to
implement the changes we sketched in Section~\ref{ss:motivation}. We
extend the grammar of multi-hole contexts~$\mhc$ (resp. $\emhc$) as follows: 

\vspace{2mm}
\noindent\begin{tabularx}{\linewidth}{@{\hspace{\parindent}}rlXr@{}}
$\mhc$ & $\bnfdef$ & $\cval \mhc \bnfor \app \mhc \mhc \bnfor \prFresh{x}{\mhc} \bnfor 
  \prReset{\cval\mhc} \mhc \bnfor \prWithSC{\cval \mhc}{x}{\mhc} \bnfor
  \prPushSC{\cval \mhc}{\mhc} \bnfor \appcont {\holecont_i} \mhc$ 
  & (contexts) \\[1mm]
$\emhc$ & $\bnfdef$ & $\mtectx \bnfor \argectx{\emhc}{\mhc} \bnfor
\valectx{\cval \mhc}{\emhc} \bnfor \prRstectx{\hole_i}{\emhc} \bnfor \appcont {\holecont_i} \emhc$
  & (evaluation contexts)
\end{tabularx}
\vspace{2mm}

\noindent The grammar of value contexts $\cval \mhc$ is unchanged. The hole
$\holecont_i$ can be filled only with a continuation; when we write
$\inctx {(\appcont {\holecont_i} \mhc)} \env$, we assume $\env_i$ is a
continuation $\reifCtx E$, and the result of the operation is
$\inctx E {\inctx \mhc \env}$ (and similarly for $\emhc$).

We also change the way we deal with captured contexts, by replacing the rule
for~$\ltsctx i \mhc$ with the two following rules---we otherwise keep unchanged
the other transitions of Figure~\ref{fig:lts}:
\vspace{1mm}
\begin{mathpar}
  \inferrule{\env_i = \reifCtx E}
  {\env \ltsctx i {\cval \mhc} (\env, \inctx E {\inctxe {\cval \mhc} \env})}
  \and
  \inferrule{\env_i = \reifCtx{\inctx{E_1}{\prReset p {E_2}}} \\ \env_j = p \\ p\notin\surPr{E_2}}
  {\env \ltsdecomp i j (\env, \reifCtx{E_1}, \reifCtx{E_2})}
\end{mathpar}
\vspace{1mm}

\noindent 
The transition $\ltsctx i {\cval \mhc}$ is the same as in
Section~\ref{s:standard}, except that it tests with an argument built
with a value context $\cval \mhc$ instead of a regular context
$\mhc$. We also introduce the transition $\ltsdecomp i j$, which
decomposes a captured context $\reifCtx{\inctx{E_1}{\prReset p
    {E_2}}}$ into sub-contexts $\reifCtx{E_1}$, $\reifCtx{E_2}$,
provided that $p$ is in $\env$. This transition is necessary to take
into account the possibility for an external observer to capture a
part of a context, scenario which can no longer be tested with
$\ltsctx i {\cval \mhc}$, as explained in Section~\ref{ss:motivation},
and as illustrated with the next example.
\begin{ex}
  Let $\env = (p, \reifCtx{\prReset p \hole})$, $\envd = (q,
  \reifCtx{\hole})$; then $\env \ltsctx 2 {\cval \mhc} (\env, \prReset
  p {\inctx {\cval \mhc} \env}) \lts\tau (\env, \inctx {\cval \mhc}
  \env)$ and $\envd \ltsctx 2 {\cval \mhc} (\envd, \inctx {\cval \mhc}
  \envd)$. Without the $\ltsdecomp i j$ transition, $\env$ and $\envd$
  would be bisimilar, which would not be sound (they are distinguished
  by the context $\prPushSC {\hole_2}{\prWithSC {\hole_1} x \Omega}$).
\end{ex}
The other rules are not modified, but
their meaning is still affected by the change in the contexts grammars: the
transitions $\ltslam i {\cval \mhc}$ and $\ltsstuck \emhc$ can now test with
more arguments. This is a consequence of the fact that an observer can build a
bigger continuation from a captured context. For instance, if $\env = (p,
\reifCtx E, \lam x {\prPushSC x v})$, then with the LTS of
Section~\ref{s:standard}, \iftoggle{paper}{we have}{} $\env \ltsctx 2 {\inctx {\emhc_1}{\prWithSC
    {\hole_1} x x}} \ltsstuck {\prReset {\hole_1}{\emhc_2}} 
\ltslam 3 {\hole_4} (\env,
\reifCtx{\inctx {\emhc_1}{\inctx E {\inctx{\emhc_2}\env}, \env}}, \prPushSC {\reifCtx{\inctx
    {\emhc_1}{\inctx E {\inctx{\emhc_2}\env}, \env}}} v)$. In the new LTS, the first transition is
no longer possible, but we can still test the $\lambda$-abstraction with the
same argument using $\env \ltslam 3 {\inctx {\emhc_1}{\appcont
    {\holecont_2}{\emhc_2}}} (\env, \prPushSC {\reifCtx{\inctx {\emhc_1}{\inctx
      E {\inctx{\emhc_2}\env}, \env}}} v)$.\\

As explained in Section~\ref{ss:motivation}, we want to prevent the use of some
up-to techniques (like the bisimulation up to related context we introduce in
Section~\ref{ss:up-to-better}) after some transitions, especially
$\ltsctx i {\cval \mhc}$. To do so, we distinguish the \emph{passive}
transitions $\ltsctx i {\cval \mhc}$, $\ltsval$ from the other ones, called
\emph{active}. A passive transition $\sts_1 \lts\act \sts_2$ can be inverted by
an up-to technique, which is possible if no new information is generated between
the states $\sts_1$ and $\sts_2$. For example, the transition
$\env \ltsval \env$ is passive, as we already know that $\env$ is composed only
of values. In contrast, the transition $\env \ltsprpt i j \env$ is active, as we
gain some information: the prompts $\env_i$ and $\env_j$ are equal.  The
transition $\env \ltsctx i {\cval \mhc} (\env, e)$ is passive at it simply
recombines existing information in $\env$ to build $e$, without any reduction
step taking place, and thus without generating new information. Some extra
knowledge is produced only when $(\env, e)$ evolves (with active transitions),
as it then tells us how the tested context~$\env_i$ actually interacts with the
value constructed from $\cval \mhc$. Finally, $\ltslam i {\cval \mhc}$ and
$\ltsstuck \emhc$ correspond to reduction steps and are therefore active, and
$\ltsdecomp i j$ is also active as it provides some information by telling us
how to decompose a continuation.


With this distinction, we change the definition of progress, to allow a relation
$\rel$ to progress towards different relations after passive and active
transitions.

\begin{defi}
  A relation $\rel$ diacritically progresses to $\rels$, $\relt$ written $\rel
  \pprogress \rels, \relt$, if $\rel \mathop\subseteq \rels$, $\rel
  \mathop\subseteq \relt$,  and $\sts \rel \stt$ implies that
  \begin{itemize}
  \item if $\sts \lts\act \sts'$ and $\lts\act$ is passive, then
    there exists $\stt'$ such that $\stt \ltswk\act \stt'$ and $\sts' \rels \stt'$;
  \item if $\sts \lts\act \sts'$ and $\lts\act$ is active, then
    there exists $\stt'$ such that $\stt \ltswk\act \stt'$ and $\sts' \relt \stt'$;
  \item the converse of the above conditions on $\stt$.
  \end{itemize}
  A $\star$-bisimulation is a relation $\rel$ such that $\rel \pprogress \rel,
  \rel$, and $\star$-bisimilarity $\bbisim$ is the union of all
  $\star$-bisimulations.
\end{defi}
With the same LTS, $\progress$ and $\pprogress$ would entail the same notions of
bisimulation and bisimilarity; the distinction between active and passive
transitions is interesting only when considering up-to techniques. We change
the notation for the bisimilarity $\bbisim$ to emphasize that we use a different LTS in
this section.

\subsection{Up-to Techniques, Soundness, and Completeness}
\label{ss:up-to-better}

We now discriminate up-to techniques, so that regular up-to techniques cannot be
used after passive transitions, while \emph{strong} ones can. An up-to technique
(resp. \emph{strong} up-to technique) is a function $f$ such that $\rel
\pprogress \rel, f(\rel)$ (resp. $\rel \pprogress f(\rel), f(\rel)$) implies
$\rel \mathop\subseteq \bbisim$. We also adapt the notions of evolution and
compatibility.

\begin{defi}
  A function $f$  evolves to $g, h$, written $f \fevolve g, h$, if
  for all $\rel \pprogress \rel, \relt$, we have $f(\rel) \pprogress g(\rel),
  h(\relt)$.

  A function $f$ \emph{strongly} evolves to $g, h$, written $f \sevolve g, h$,
  if for all $\rel \pprogress \rels, \relt$, we have $f(\rel) \pprogress
  g(\rels), h(\relt)$.
\end{defi}
Strong evolution is very general, as it uses any relation $\rel$, while regular
evolution is more restricted, as it relies on relations $\rel$ such that $\rel
\pprogress \rel, \relt$. But the definition of \emph{diacritical compatibility}
below still allows to use any combinations of strong up-to techniques after a
passive transition, even for functions which are not themselves strong. In
contrast, regular functions can only be used once after a passive transition of
an other regular function.

\begin{defi}
  A set $\setF$ of continuous functions is \emph{diacritically compatible} if there
  exists $\setS \subseteq \setF$ such that
  \begin{itemize}
  \item for all $f \in \setS$, we have $f \sevolve {\fid \setS}^\omega, {\fid \setF}^\omega$;
  \item for all $f \in \setF$, we have $f \fevolve {\fid \setS}^\omega \comp
    {\fid \setF} \comp {\fid \setS}^\omega, {\fid \setF}^\omega$.
  \end{itemize}
\end{defi}
If $(\setS_i)_{i \in I}$ is a family of subsets of $\setF$ which verify the
conditions of the definition, then $\bigcup_{i \in I} \setS_i$ also verifies
them. We can therefore consider the largest of such subsets, written
$\strong \setF$, which can be defined as the union of all subsets of $\setF$
verifying the conditions of the definition. This (possibly empty) subset of
$\setF$ contains the strong up-to techniques of~$\setF$.

\begin{lem}
  \label{l:properties-compatibility-better}
  Let $\setF$ be a diacritically compatible set.
  \begin{itemize}
  \item If $\rel \pprogress {\widehat{\mathsf{strong}(\setF)}}^\omega(\rel),
    {\fid \setF}^\omega(\rel)$,
    then ${\fid \setF}^\omega(\rel)$ is a $\star$-bisimulation.
  \item If $f \in \setF$,  then $f$ is an up-to technique. If $f \in \strong \setF$, then
    $f$ is a strong up-to technique.
  \item For all $f \in \setF$, we have $f(\bbisim) \mathop\subseteq \bbisim$.
  \end{itemize}
\end{lem}

\begin{proof}
  Let $\setS \defeq {\widehat{\mathsf{strong}(\setF)}}$. For the first item, we
  prove that for all $n$
  \[(\setS^\omega \comp {\fid \setF} \comp \setS^\omega)^n(\rel) \pprogress
  (\setS^\omega \comp {\fid \setF} \comp \setS^\omega)^n(\setS^\omega(\rel)),
  {\fid \setF}^\omega(\rel)\]
  by induction on~$n$. There is nothing to prove for $n=0$. Suppose
  $n>0$. We know that
  \[(\setS^\omega \comp {\fid \setF} \comp \setS^\omega)^{n-1}(\rel) \pprogress
  (\setS^\omega \comp {\fid \setF} \comp \setS^\omega)^{n-1}(\setS^\omega(\rel)),
  {\fid \setF}^\omega(\rel).\] For all $f \in \setS$, we have
  \[f((\setS^\omega \comp {\fid \setF} \comp \setS^\omega)^{n-1}(\rel)) \pprogress
  \setS^\omega(\setS^\omega \comp {\fid \setF} \comp
  \setS^\omega)^{n-1}(\setS^\omega(\rel)), {\fid \setF}^\omega({\fid \setF}^\omega(\rel)),\]
  therefore we have
  \[\setS^\omega(((\setS^\omega \comp {\fid \setF} \comp \setS^\omega)^{n-1}(\rel)))
  \pprogress \setS^\omega(\setS^\omega \comp {\fid \setF} \comp
  \setS^\omega)^{n-1}(\setS^\omega(\rel)), {\fid \setF}^\omega(\rel).\]
  Because
  $\setS^\omega \comp (\setS^\omega \comp {\fid \setF} \comp \setS^\omega)^{n-1} =
  \setS^\omega \comp (\setS^\omega \comp {\fid \setF} \comp \setS^\omega)^{n-1} \comp
  \setS^\omega$, for all $f \in {\fid \setF}$, we have
  \[f(\setS^\omega(((\setS^\omega \comp {\fid \setF} \comp \setS^\omega)^{n-1}(\rel)))
  \pprogress (\setS^\omega \comp {\fid \setF} \comp
  \setS^\omega)(\setS^\omega(\setS^\omega \comp {\fid \setF} \comp
  \setS^\omega)^{n-1}(\setS^\omega(\rel))), {\fid \setF}^\omega({\fid \setF}^\omega(\rel)),\]
  which implies
  ${\fid \setF}(\setS^\omega(((\setS^\omega \comp {\fid \setF} \comp
  \setS^\omega)^{n-1}(\rel))) \pprogress (\setS^\omega(\setS^\omega \comp {\fid \setF}
  \comp \setS^\omega)^n(\setS^\omega(\rel))), {\fid \setF}^\omega(\rel)$.
  Finally, composing again with $\setS^\omega$, we obtain
  \[\setS^\omega(((\setS^\omega \comp {\fid \setF} \comp \setS^\omega)^n(\rel)) \pprogress \setS^\omega
  \comp (\setS^\omega \comp {\fid \setF} \comp \setS^\omega)^n(\setS^\omega(\rel)), {\fid \setF}^\omega(\rel),\]
  as wished. 

  Because ${\fid \setF}^\omega = (\setS^\omega \comp {\fid \setF} \comp \setS^\omega)^\omega$,
  we get that
  ${\fid \setF}^\omega(\rel) \pprogress {\fid \setF}^\omega(\rel), {\fid \setF}^\omega(\rel)$, i.e.,
  ${\fid \setF}^\omega(\rel)$ is a $\star$-bisimulation.

  For the second item, let $f \in \setF$ and $\rel \pprogress \rel, f(\rel)$.
  Then $\mathord{\rel} \subseteq {\fid \setF}^\omega(\rel)$ by definition of $\omega$ and
  ${\fid \setF}^\omega(\rel) \mathop\subseteq \bbisim$ by the first item. Therefore we
  have $\rel \mathop\subseteq \bbisim$ and $f$ is an up-to technique. Similarly,
  we can show that $f \in \strong{\setF}$ and $\rel \pprogress f(\rel), f(\rel)$
  imply $\rel \mathop\subseteq \bbisim$, meaning that $f$ is a strong up-to technique. 

  For the last item, for all $f \in \setF$, we have
  $f(\bbisim) \subseteq {\fid \setF}^\omega(\bbisim)$, and ${\fid \setF}^\omega(\bbisim)
  \mathop\subseteq \bbisim$ by the first item, so we have $f(\bbisim) \mathop\subseteq
  \bbisim$ as wished.
\end{proof}

We now use this framework to define up-to techniques for the
$\star$-bisimulation. The definitions of $\rawperm$ and $\rawweak$ are
unchanged. We define bisimulation up to related contexts for values
$\rawutrctxv$ and for any term $\rawutrctx$ as follows:
\vspace{1mm}
\begin{mathpar}
\inferrule{\env \rel \envd}
{(\env, \inctx{\mhcvs}{\env}, \inctx{\mhc}{\env}) \utrctxv\rel 
 (\envd,\inctx{\mhcvs}{\envd},\inctx{\mhc}{\envd})}
\and
\hspace{-1em}\inferrule{(\env, e_1) \rel (\envd, e_2)}
{(\env, \inctx{\mhcvs}{\env}, \inctx{\emhc}{e_1, \env}) \utrctx\rel
 (\envd,\inctx{\mhcvs}{\envd},\inctx{\emhc}{e_2, \envd})}
\end{mathpar}
\vspace{1mm}

\noindent
The definitions look similar to the ones of $\rawutctxv$ and
$\rawutctx$, but the grammar of multi-hole contexts now include
$\holecont_i$. Besides, we inline strengthening in the definitions of
$\rawutrctxv$ and $\rawutrctx$, allowing $\env$, $\envd$ to be
extended. This is necessary because, e.g., $\rawstr$ and $\rawutrctx$
cannot be composed after a passive transition (they are both not
strong), so $\rawutrctx$ have to include $\rawstr$ directly. Note that
the behavior of $\rawstr$ can be recovered from $\rawutrctx$ by taking
$\emhc = \hole$.

\begin{lem} $\setF \defeq \{ \rawperm, \rawweak, \rawutrctxv, \rawutrctx\}$ is
  diacritically compatible, with $\strong \setF= \{\rawperm, \rawweak\}$. 
\end{lem}
As a result, these functions are up-to techniques, and $\rawweak$ and
$\rawperm$ can be used after a passive transition. Because of the last
item of Lemma~\ref{l:properties-compatibility-better}, $\bbisim$ is
also a congruence w.r.t. evaluation contexts, which means that
$\bbisim$ is sound w.r.t. $\ectxEq$. We can also prove it is complete
the same way as for Theorem~\ref{t:charac-std}, leading again to full
characterization.
\begin{thm}
  $e_1 \ectxEq e_2$ iff $(\emptyset, e_1) \bbisim (\emptyset, e_2)$.
\end{thm}

\begin{rem}
  \label{r:stuck-error-better}
  If we consider control-stuck terms as errors, as in
  Remark~\ref{r:stuck-error}, then we can use the transition of
  Remark~\ref{r:stuck-error-standard}, considered as active, and the results of
  this section scale to such a version of the bisimilarity. While the
  compatibility proof for $\rawutrctx$ does not change much, the one for
  $\rawutrctxv$ needs an extra case analysis to deal with the modified
  $\ltsstuck \emhc$ transition; see \iftoggle{paper}{
    \cite[Remark
    B.3]{Aristizabal-al:Inria16}}{Remark~\ref{rem:modified-stuck-proof}}
  for further details.
\end{rem}

\subsection{Examples}

We illustrate the use of $\bbisim$, $\rawutrctxv$, and $\rawutrctx$ with two
examples that would be much harder to prove with the techniques of
Section~\ref{s:standard}.

\begin{ex}[$\beta_\Omega$ axiom]
  \label{ex:beta-omega}
  We prove $\app{(\lam x {\inctx E x})}{e} \bbisim \inctx E e$ if $x \notin
  \freeVars E$ and $\surPr{E} = \emptyset$. Define $\rel$ starting with 
  $(\reifCtx{\app{(\lam x {\inctx E x})} \hole}) \rel (\reifCtx E)$,
  and closing it under the $(\promptCheckRule)$ and the following rule:
  \begin{mathpar}
    \inferrule{\env \rel \envd}
    {(\env, \app{(\lam x {\inctx E x})}{\inctx {\cval \mhc} \env})
    \rel
    (\envd, \inctx E {\inctx {\cval \mhc} \envd})}
  \end{mathpar}
  \vspace{1mm}

  \noindent
  Then $(\emptyset, \app{(\lam x {\inctx E x})}{e}) 
  \weak{\utrctxv\rel} (\emptyset, \inctx E e)$
  and $\rel$ is a bisimulation up to context, since
  the sequence $\env \ltsctx 1 {\cval \mhc} (\env, \app{(\lam x {\inctx E
      x})}{\inctx {\cval \mhc} \env}) \lts\tau (\env, \inctx E {\inctx {\cval
      \mhc} \env})$ fits
    $\envd \ltsctx 1 {\cval \mhc} (\envd, \inctx E {\inctx
    {\cval \mhc} \envd}) \ltswk\tau (\envd, \inctx E {\inctx {\cval \mhc}
    \envd})$,
  where the final states are in $\rawutrctxv$. Notice we use $\rawutrctxv$
  after $\lts\tau$, and not after the passive $\ltsctx 1 {\cval \mhc}$
  transition.
\end{ex}

\begin{ex}[Exceptions]
A possible way of extending a calculus with exception handling is to
add a construct $\mathsf{try}_\var{r}\;e\;\mathsf{with}\;v$, which
evaluates $e$ with a function raising an exception stored under the
variable~$\var{r}$. When~$e$ calls the function in~$r$ with some
argument~$v'$, even inside another $\mathsf{try}$ block, then the
computation of $e$ is aborted and replaced by $\app{v}{v'}$.  We can
implement this behavior directly in $\lambdabla$; more precisely, we
write $\mathsf{try}_\var{r}\;e\;\mathsf{with}\;v$ as
$\app{\app{\mathsf{handle}}{(\lam{r}{e})}}{v}$, where
$\mathsf{handle}$ is a function expressed in the calculus. One
possible implementation of $\mathsf{handle}$ in $\lambdabla$ is very
natural and heavily relies on fresh-prompt generation:
\begin{eqnarray*}
\mathsf{handle} &\defeq&
    \lam{f}{\lam{h}{\prFresh{x}{\prReset{\var{x}}{\app{\var{f}}{(\lam{z}{
        \prWithSC{\var{x}}{\osef}{\app{\var{h}}{\var{z}}}})}}}}}
\end{eqnarray*}
The idea is to raise an exception by aborting the current continuation
up to the corresponding prompt. The same function can be implemented
using any comparable-resource generation and only one prompt $p$:
\begin{eqnarray*}
\mathsf{handle}_p &\defeq&
    \lam{f}{\lam{h}{\prFresh{x}{
        \app{\app{(\prReset{\var{p}}{
            \expLet{r}{\app{\var{f}}{\mathsf{raise}_{p,\var{x}}}}{
            \lam{\osef}{\lam{\osef}{\var{r}}}}})}{
        \var{x}}}{\var{h}}}}}\\
\mathsf{raise}_{p,\var{x}} &\defeq&
    \expFix{r}{z}{\prWithSC{\var{p}}{\osef}{\lam{y}{\lam{h}{
        \expIf{\prEqual{\var{x}}{\var{y}}}{
            \app{\var{h}}{\var{z}}}{
            \app{\var{r}}{\var{z}}}}}}}
\end{eqnarray*}
Here the idea is to keep a freshly generated name $x$ and a handler
function $h$ with the prompt corresponding to each call of
$\mathsf{handle}_p$. The exception-raising function
$\mathsf{raise}_{p,\var{x}}$ iteratively aborts the current delimited
continuation up to the nearest call of $\mathsf{handle}_p$ and checks
the name stored there in order to find the corresponding handler. Note
that this implementation also uses prompt generation, since it is the
only comparable resource that can be dynamically generated in
$\lambdabla$, but the implementation can be easily translated to,
e.g., a calculus with single-prompted delimited-control operators and
first-order store.
\end{ex}
\begin{proof}
  We prove that both versions of $\mathsf{handle}$ are $\star$-bisimilar.  As in
  Example~\ref{ex:folklore} we iteratively build a relation $\rel$ closed
  under the $(\promptCheckRule)$ rule, so that
  $\rel$ is a bisimulation up to context. We start with $(\mathsf{handle})
  \rel (\mathsf{handle}_p)$; to match the $\ltslam{1}{\cval\mhc}$ transition, we
  extend $\rel$ as follows:
\vspace{1mm}
\begin{mathpar}
\inferrule{\env \rel \envd}
{(\env, \lam{h}{\prFresh{x}{
    \prReset{\var{x}}{\app{\inctxe{\cval\mhc}\env}{(\lam{z}{
        \prWithSC{\var{x}}{\osef}{\app{\var{h}}{\var{z}}}})}}}})
\rel
(\envd, \lam{h}{\prFresh{x}{
    \app{\app{(\prReset{\var{p}}{
        \expLet{r}{\app{\inctxe{\cval\mhc}\envd}{\mathsf{raise}_{p,\var{x}}}}{
        \lam{\osef}{\lam{\osef}{\var{r}}}}})}{
    \var{x}}}{\var{h}}}})
}
\end{mathpar}
\vspace{1mm}

\noindent
We obtain two functions which are in turn tested with
$\ltslam{n+1}{\cval{\mhc'}}$, and we obtain the states
\[
(\env, \prReset{p_1}{\app{\inctxe{\cval\mhc}\env}{(\lam{z}{
    \prWithSC{p_1}{\osef}{\app{\inctxe{\cval{\mhc'}}\env}{\var{z}}}})}}) \mbox{ and
}
(\envd, \app{\app{(\prReset{p}{
        \expLet{r}{\app{\inctxe{\cval\mhc}\envd}{\mathsf{raise}_{p,p_2}}}{
        \lam{\osef}{\lam{\osef}{\var{r}}}}})}{
    p_2}}{\inctxe{\cval{\mhc'}}\envd}).
\]
Instead of adding them to $\rel$ directly, we decompose them into
corresponding parts using up to context (with $\mhc =
\appcont{\holecont_{n+1}}{\app{\cval\mhc}{\hole_{n+2}}}$), and we add
these subterms to $\rel$:
\vspace{1mm}
\[
\inferrule{\env \rel \envd
    \and p_1 \notin \promptsOf{\env}
    \and p_2 \notin \promptsOf{\envd}}
{(\env, 
\reifCtx{\prReset{p_1}{\mtectx}},
\lam{z}{\prWithSC{p_1}{\osef}{\app{\inctxe{\cval{\mhc'}}{\env}}{\var{z}}}})
\rel
(\envd, 
\reifCtx{\app{\app{
    (\prReset{p}{\expLet{r}{\mtectx}{\lam{\osef}{\lam{\osef}{\var{r}}}}})
    }{p_2}}{\inctxe{\cval{\mhc'}}\envd}},
\mathsf{raise}_{p,p_2})
} \tag{$**$}
\]
\vspace{1mm}

\noindent
Testing the two captured contexts with $\ltsctx{n+1}{\cval{\mhc''}}$
is easy, because they both evaluate to the thrown value.  We now
consider
$\lam{z}{\prWithSC{p_1}{\osef}{\app{\inctxe{\cval{\mhc'}}{\env}}{\var{z}}}}$
and $\mathsf{raise}_{p,p_2}$; after the transition
$\ltslam{n+2}{\cval\mhc}$ we get the two control stuck terms
\begin{mathpar}
\prWithSC{p_1}{\osef}{\app{\inctxe{\cval{\mhc'}}{\env}}{\inctxe{\cval\mhc}}{\env}}
\quad\textrm{and}\quad
\prWithSC{p}{\osef}{\lam{y}{\lam{h}{
    \expIf{\prEqual{p_2}{\var{y}}}{
        \app{\var{h}}{\inctxe{\cval\mhc}{\envd}}}{
        \app{\mathsf{raise}_{p,p_2}}{\inctxe{\cval\mhc}{\envd}}}}}}
\end{mathpar}
Adding such terms to the relation will not be enough.
The first one can be unstuck only using the corresponding context 
$\reifCtx{\prReset{p_1}{\mtectx}}$, but the second one can be unstuck using
any context added by rule $(*)$, even for a different $p_2$.
In such a case, it will consume a part of the context and evaluate to itself.
To be more general we add the following rule:
\vspace{1mm}
\begin{mathpar}
\inferrule{\env \rel \envd
  \and 
  \inctxe{\emhc}{
    \prWithSC{p_1}{\osef}{\app{\inctxe{\cval{\mhc'}}{\env}}{\inctxe{\cval\mhc}}{\env}}
    ,\env} \textrm{ is control-stuck}
}
{(\env,
  \inctxe{\emhc}{
    \prWithSC{p_1}{\osef}{\app{\inctxe{\cval{\mhc'}}{\env}}{\inctxe{\cval\mhc}}{\env}}
    ,\env})
\rel
(\envd,
  \prWithSC{p}{\osef}{\lam{y}{\lam{h}{
    \expIf{\prEqual{p_2}{\var{y}}}{
        \app{\var{h}}{\inctxe{\cval\mhc}{\envd}}}{
        \app{\mathsf{raise}_{p,p_2}}{\inctxe{\cval\mhc}{\envd}}}}}})}
\end{mathpar}
\vspace{1mm}

\noindent
The newly introduced stuck terms are tested with $\ltsstuck{\emhc'}$;
if $\emhc'$ does not have $\holecont_i$ surrounding~$\hole$, they are
still stuck, and we can use up to evaluation context to conclude.
Assume $\emhc' = \inctxe{\emhc_1}{\inctxe{\holecont_i}{\emhc_2}}$
where $\emhc_2$ has not $\holecont_j$ around $\hole$.
If $i$ points to the evaluation context added by~$(**)$ for the same $p_2$,
then they both evaluate to terms of the same shape, so we use up to context
with $\mhc = \inctxe{\emhc_1}{\app{\cval\mhc'}{\cval\mhc}}$.
Otherwise, we know the second program compares two different prompts,
so it evaluates to
$\inctxe{\emhc_1}{\prWithSC{p}{\osef}{\lam{y}{\lam{h}{
        \expIf{\prEqual{p_2}{\var{y}}}{
          \app{\var{h}}{\inctxe{\cval\mhc}{\envd}}}{
          \app{\mathsf{raise}_{p,p_2}}{\inctxe{\cval\mhc}{\envd}}}}}},\envd}$
and we use $\rawutrctx$ with the last rule.
\end{proof}

\section{Shift and Reset}
\label{s:shift-reset}

In this section, we show how $\star$-bisimilarity can be defined for
$\lamshift$, a $\lambda$-calculus extended with \textshift and
\textreset. These operators can be encoded in $\lambdabla$ (see
Example~\ref{ex:folklore}), but relying on this encoding would lead to
a sound, but not complete bisimilarity for \textshift and
\textreset. Indeed, there are terms equivalent in $\lamshift$, the
encodings of which are no longer equivalent with the more expressive
constructs of $\lambdabla$: see
Example~\ref{ex:lambdaS-lambdabla}. This is why we work
with~$\lamshift$ is this section, and not $\lambdabla$.

We study several bisimilarities for $\lamshift$ in previous
works~\cite{Biernacki-Lenglet:FOSSACS12,Biernacki-Lenglet:FLOPS12,Biernacki-Lenglet:APLAS13,Biernacki-al:HAL15}. In
particular, we define environmental ones
in~\cite{Biernacki-Lenglet:APLAS13,Biernacki-al:HAL15}, but without a
relation equivalent to bisimulation up to related contexts, which
makes the proof of the $\beta_\Omega$ axiom very difficult in these
papers. The proof in Example~\ref{ex:beta-omega} is as easy as the
proof of the $\beta_\Omega$ axiom in~\cite{Biernacki-Lenglet:FLOPS12},
but the bisimilarity of~\cite{Biernacki-Lenglet:FLOPS12} is not
complete. Therefore, a sound and complete $\star$-bisimilarity for
$\lamshift$ which allows for simple equivalence proofs thanks to up-to
techniques improves over our previous work.

\subsection{Syntax, Semantics, and Contextual Equivalence}

The calculus $\lamshift$ is a single-prompted version of $\lambdabla$,
where the now unique delimiter $\rawreset$ is called \textreset and
the capturing construct~$\rawshift$ is called \textshift. The syntax
of the different entities is as follows.

\vspace{2mm}
\noindent
\begin{tabularx}{\textwidth}{@{\hspace{\parindent}}rlXr@{}}
  $e$ & $\bnfdef$ & $v \bnfor \app{e}{e} \bnfor \reset{e} \bnfor \shift x e$ & (expressions) \\[1mm]
  $v$ & $\bnfdef$ & $\var{x} \bnfor \lam{x}{e}$
  & (values) \\[1mm]
  $E$ & $\bnfdef$ & $\mtectx \bnfor \argectx{E}{e} \bnfor \valectx{v}{E}$
  & (pure contexts)\\[1mm]
  $F$ & $\bnfdef$ & $\mtectx \bnfor \argectx{F}{e} \bnfor \valectx{v}{F}
  \bnfor \reset{F}$
  & (evaluation contexts)
\end{tabularx}
\vspace{2mm} 

\noindent We distinguish two kinds of evaluation contexts: pure
contexts, ranged over by $E$, can be captured by \textshift, while
those represented by $F$ are the regular evaluation contexts. Captured
contexts are no longer part of the syntax, but are instead turned into
$\lambda$-abstractions, as we can see in the following reduction
rules.

\vspace{1mm}
\begin{minipage}{0.55\textwidth}
\[
\begin{array}{rcl}
\app{(\lam{x}{e})}{v} & \rawred & \subst{e}{x}{v} \\[1mm]
\reset{v} & \rawred & v \\[1mm]
\reset{\inctx{E}{\shift x {e}}} & \rawred &
  \reset{\subst{e}{x}{\lam y {\reset {\inctx E y}}}} \quad y \mbox{ fresh}
\end{array}
\]
\end{minipage}
\begin{minipage}{0.45\textwidth}

\begin{mathpar}
  \inferrule[Compatibility]{\red{e_1}{e_2}}
            {\red{\inctx{F}{e_1}}{\inctx{F}{e_2}}}
\end{mathpar}
\end{minipage}
\vspace{1mm}

\noindent The operator $\rawshift$ captures a surrounding context $E$ up to the
first enclosing \textreset. This \textreset is left in place, but $E$ remains
delimited when captured in $\lam y {\reset {\inctx E y}}$.

The original semantics of \textshift and
\textreset~\cite{Biernacka-al:LMCS05} applies these rules only to
terms with an outermost \textreset; this requirement is often lifted
in practical implementation~\cite{Dybvig-al:JFP06,Filinski:POPL94} or
studies of these
operators~\cite{Asai-Kameyama:APLAS07,Kameyama:HOSC07}. As
in~\cite{Biernacki-Lenglet:APLAS13,Biernacki-al:HAL15}, we define
equivalences for the original and the relaxed semantics. The two
semantics differ mainly in the normal forms they produce: an
expression~$\reset e$ cannot reduce to a control-stuck term $\inctx E
{\shift x {e'}}$ in the original semantics, while such a normal form
can still be obtained with the relaxed semantics. As a result, we
distinguish the observable actions for the original semantics
$\eqObsO$ from those for the relaxed semantics $\eqObsR$. Unlike
in~$\lambdabla$, both semantics cannot produce errors, so we simply
write $e \diverge$ when $e$ diverges.

\begin{defi}
We write $e_1 \eqObsO e_2$ if
\begin{enumerate}
\item $\redrtc{e_1}{v_1}$ iff $\redrtc{e_2}{v_2}$,
\item $e_1 \diverge$ iff $e_2 \diverge$.
\end{enumerate}
We write $e_1 \eqObsR e_2$ if
\begin{enumerate}
\item $\redrtc{e_1}{v_1}$ iff $\redrtc{e_2}{v_2}$,
\item $\redrtc{e_1}{\inctx {E_1}{\shift x {e'_1}}}$ iff
  $\redrtc{e_2}{\inctx {E_2}{\shift x {e'_2}}}$,
\item $e_1 \diverge$ iff $e_2 \diverge$.
\end{enumerate}
\end{defi}
\noindent Similarly, we define a contextual equivalence for each semantics. 

\begin{defi}[Contextual equivalence]
  Given two closed expressions $e_1$ and $e_2$, we write $e_1 \ctxEqO e_2$ if
  for all $E$, we have $\reset{\inctx{E}{e_1}} \eqObsO \reset{\inctx{E}{e_2}}$,
  and we write $e_1 \ctxEqR e_2$ if for all $E$, we have
  $\inctx{E}{e_1} \eqObsR \inctx{E}{e_2}$.
\end{defi}

\noindent Because we no longer have resource generation, note that
testing with evaluation contexts~$F$ is equivalent to testing with any
context $C$ in $\lamshift$~\cite{Biernacki-al:HAL15}.

\begin{ex}
  \label{ex:lambdaS-lambdabla} The expressions
  $\reset{\reset {e_1} \iapp (\reset {e_2} \iapp {\shift x {\lam y
        y}})}$ and $\reset{\reset {e_2} \iapp (\reset {e_1} \iapp
    {\shift x {\lam y y}})}$ are contextually equivalent in
  $\lamshift$ with either semantics, but their encodings are not
  bisimilar in $\lambdabla$. In $\lamshift$, depending on whether
  $\reset{e_1}$ or $\reset{e_2}$ diverge or reduce to a value, the two
  above terms either diverge or reduce to $\lam y y$. In $\lambdabla$,
  the encoding of $\reset{e_1}$ can reduce to a control-stuck term,
  e.g., if $e_1 = \prFresh x {\prWithSC x y y}$, making $\reset{\reset
    {e_1} \iapp (\reset {e_2} \iapp {\shift x {\lam y y}})}$ stuck as
  well, while $e_2$ may diverge, and a stuck term is not equivalent to
  a diverging one.

\end{ex}

\begin{rem}
  \label{rem:throw}
  We can equivalently define $\lamshift$ with captured pure contexts
  as values and a throw construct $\prPushSC v t$, as in~$\lambdabla$,
  using the following reduction rules
  \vspace{1mm}
  \[
  \begin{array}{rcl}
    \shift x e & \rawred & \subst{e}{x}{\reifCtx E} \\[1mm]
    \prPushSC {\reifCtx E} v & \rawred &  \inctx E v
  \end{array}
  \]
  \vspace{1mm}

  \noindent
  and with $\prPushSC {\reifCtx E} F$ as an evaluation context and
  $\prPushSC {\reifCtx E}{E'}$ as a pure context. Only values are
  thrown to captured contexts, unlike in $\lambdabla$. In this
  section, we stick to the syntax we use
  in~\cite{Biernacki-Lenglet:APLAS13,Biernacki-al:HAL15} to facilitate
  comparisons with these papers. We discuss how to adapt the LTS to
  the syntax with throw in Remark~\ref{rem:lts-throw}.
\end{rem}

\subsection{Bisimilarity and Up-to Techniques} 

For bisimulation up to related contexts to be useful, we want to be
able to save evaluation context (not necessarily pure) in states. To
do so, we let $\enve$, $\envef$ range over sequences of evaluation
contexts, and we consider states of the form $(\enve, \env, e)$, where
$\env$ is still a sequence of values. Multi-hole contexts, whose
syntax is given below, are now filled with~$\enve$ and $\env$.

\vspace{2mm}
\noindent\begin{tabularx}{\linewidth}{@{\hspace{\parindent}}rlXr@{}}
$\mhc$ & $\bnfdef$ & $\cval \mhc \bnfor \app \mhc \mhc \bnfor 
  \reset \mhc \bnfor \shift{x}{\mhc} \bnfor \appcont {\holecont_i} \mhc$ 
  & (contexts) \\[1mm]
$\cval \mhc$ & $\bnfdef$ & $\var{x} \bnfor \lam{x}{\mhc} \bnfor \hole_i$
  & (value contexts) \\[1mm]
$\fmhc$ & $\bnfdef$ & $\mtectx \bnfor \argectx{\fmhc}{\mhc} \bnfor
\valectx{\cval \mhc}{\fmhc} \bnfor \reset{\fmhc} \bnfor \appcont {\holecont_i} \fmhc $
  & (evaluation contexts)
\end{tabularx}
\vspace{2mm}

\noindent We write $\inctx \mhc {\enve, \env}$ to say that $\holecont_i$ of
$\mhc$ is filled with the context $\enve_i$, as in Section~\ref{s:better}, and
each hole $\hole_j$ is plugged with the value $\env_j$. As before, it assumes
that each index $i$ of $\holecont_i$ is smaller than the size of $\enve$, and
each $j$ of $\hole_j$ is smaller than the size of $\env$. Similarly, we write
$\inctx \fmhc {e, \enve, \env}$ for evaluation contexts, so that $e$ goes into
$\hole$.

\begin{figure}
\vspace{1mm}
\textbf{Rules common to both semantics:}
\begin{mathpar}
  \inferrule{\red{e_1}{e_2}}
  {(\enve, \env, e_1) \lts\tau (\enve, \env, e_2)}
  \and
  \inferrule{\env_j = \lam x e}
  {(\enve, \env) \ltslam j {\cval\mhc} (\enve, \env, \subst e x
    {\inctxe{\cval\mhc}{\enve, \env}})}
  \and
  \inferrule{ }
  {(\enve, \env) \ltsval (\enve, \env)}
  \and
  \inferrule{e \mbox{ is stuck} \\ \inctx \fmhc {e, \enve, \env} \redbis e'}
  {(\enve, \env, e) \ltsstuck \fmhc (\enve, \env, e')}
  \and
  \inferrule{ }
  {(\enve, \env) \ltsctxx i {\cval\mhc} (\enve, \env, \inctx {\enve_i} {\inctxe
      {\cval\mhc}{\enve, \env}})}
  \and
  \inferrule{\enve_i = E}
  {(\enve, \env) \ltspure i (\enve, \env)}
  \and
  \inferrule{\enve_i = \inctx F {\reset E}}
  {(\enve, \env) \ltsnotpure i (\enve, \inctx F {\reset \hole}, \reset E, \env)}
\vspace{2mm}
\end{mathpar}

\textbf{Extra rule for the original semantics:}
\begin{mathpar}
  \inferrule{e \mbox{ is not stuck} }
  {(\enve, \env, e) \ltsstuck \fmhc (\enve, \env, \inctx  \fmhc {e, \enve, \env})}
\vspace{2mm}
\end{mathpar}

\textbf{Up-to techniques for both semantics:}
\begin{mathpar}
  \inferrule{(\seq {F}, \enve, \seq v, \env, e_1) \rel (\seq {
      {F'}}, \envef, \seq w, \envd, e_2)} {(\enve, \env, e_1) \weak\rel
    (\envef, \envd, e_2)} 
  \and 
\inferrule{(\enve, \env) \rel (\envef, \envd)}
{(\enve, \inctx \fmhcs {\enve, \env}, \env, \inctx{\mhcvs}{\enve, \env}, \inctx{\mhc}{\enve, \env}) \utrctxv\rel 
 (\envef, \inctx \fmhcs {\envef, \envd}, \envd,\inctx{\mhcvs}{\envef, \envd},\inctx{\mhc}{\envef, \envd})}
\and
\inferrule{(\enve, \env, e_1) \rel (\envef, \envd, e_2)}
{(\enve, \inctx \fmhcs {\enve, \env}, \env, \inctx{\mhcvs}{\enve, \env}, \inctx{\fmhc}{e_1, \enve, \env}) \utrctx\rel
 (\envef, \inctx \fmhcs {\envef, \envd}, \envd,\inctx{\mhcvs}{\envef, \envd},\inctx{\fmhc}{e_2, \envef, \envd})}
\end{mathpar}

\caption{LTS and up-to techniques for \textshift and \textreset}
\label{fig:LTS-shift}
\end{figure}

We present the LTS and up-to techniques for the two semantics of $\lamshift$ in
Figure~\ref{fig:LTS-shift}. In $\lambdabla$, having $\star$ holes in multi-hole
contexts helps when testing captured contexts as well as for the up-to
techniques. In contrast, in $\lamshift$, $\star$ holes are useful only for the
up-to techniques, and not for the bisimilarity itself, even if we consider the
syntax with captured contexts (see Remark~\ref{rem:lts-throw}). As a result,
some of the transitions are only for the bisimilarity, namely $\lts\tau$,
$\ltslam j {\cval \mhc}$, $\ltsval$, and $\ltsstuck{\fmhc}$, while the remaining
three are for bisimulations up to context: they are used only if $\enve$ is not
empty.

The transition $\ltsctxx i {\cval\mhc}$ tests the evaluation context $\enve_i$
by passing it a value built from $\enve$ and $\env$. A stuck term is able to
distinguish a pure context from an impure one, and it can extract from
$\inctx F {\reset E}$ the context up to the first enclosing reset $\reset E$.
However, unlike in $\lambdabla$, we cannot decompose $F$ further, because the
capture leaves the delimiter in place: we can distinguish $\hole$ from
$\reset\hole$, but not $\reset \hole$ from $\reset{\reset \hole}$. We use
$\ltspure i$ and $\ltsnotpure i$ to perform these tests: $\ltspure i$ simply
states that $\enve_i$ is pure, while $\ltsnotpure i$ decomposes
$\enve_i = \inctx F {\reset E}$ into $\inctx F {\reset \hole}$ and $\reset E$.
Because we leave a \textreset inside $F$, applying $\ltsnotpure i$ to
$\inctx F {\reset \hole}$ does not decompose~$F$ further, but simply generates
$\inctx F {\reset \hole}$ again (and $\reset \hole$), and duplicated contexts
can then be ignored thanks to strengthening.

The transition $\ltsstuck \fmhc$ compares stuck terms in the relaxed
semantics. In the original semantics, we can also relate with the extra rule a
stuck term with a regular term: we prove in Example~\ref{ex:skkt} that
$\shift k {\app k e}$ is equivalent to $e$ in that semantics if
$k \notin \freeVars e$. When the extra rule is applied to two non stuck terms
$e_1$ and $e_2$, it generates expressions $\inctx \fmhc {e_1, \enve, \env}$ and
$\inctx \fmhc {e_2, \envef, \envd}$ which are automatically related with up to
contexts, so the extra rule does not produce additional testing for regular
terms. The transition $\ltsstuck \fmhc$ uses any evaluation context~$\fmhc$, and
not simply a context of the form $\reset\emhc$ with $\emhc$ a pure context, as
we do in~\cite{Biernacki-Lenglet:APLAS13, Biernacki-al:HAL15}. We do so to take
$\holecont_i$ into account: a context $\appcont {\holecont_i} \emhc$ may also
trigger a capture if $\enve_i$ is an impure context. Besides, if
$(\enve, \env, e_1) \rel (\envef, \envd, e_2)$ and $\enve_i$ is pure, then
$\envef_i$ may be impure if~$e_1$ and~$e_2$ contain infinite behavior (and
thus, the transitions $\ltspure i$ and $\ltsnotpure i$ are never applied). For
example, we have
$(\hole, \emptyset, \shift k \Omega) \ltsstuck {\appcont{\holecont_1} \hole}
(\hole, \emptyset, \shift k \Omega)$
and
$(\reset\hole, \emptyset, \shift k \Omega) \ltsstuck {\appcont{\holecont_1}
  \hole} (\reset\hole, \emptyset, \reset \Omega)$;
the two resulting states are distinguished in the relaxed semantics, but they
are equated in the original one. However, what is beyond the first enclosing
\textreset of a testing context $\inctx \fmhc {\enve, \env}$, and therefore do
not interact with the tested terms, can be ignored thanks to bisimulation up to
related contexts, as in Example~\ref{ex:stuck-utctx}.

The transitions $\lts\tau$, $\ltslam j {\cval \mhc}$, and $\ltsstuck{\fmhc}$ are
active because they correspond to reduction steps, and $\ltspure i$ and
$\ltsnotpure i$ are active because they provide information on the tested
contexts (being pure or not, and how to decompose contexts that are not
pure). As before, $\ltsval$ is passive because it informs about the nature of
the tested states (composed only of values), and $\ltsctxx i {\cval\mhc}$ is
passive because it does not provide any information on the tested context nor
does it correspond to a reduction step. 

\begin{rem}
  \label{rem:lts-throw}
  If captured contexts are considered values, as suggested in
  Remark~\ref{rem:throw}, then they are stored in $\env$ and $\envd$,
  and therefore cannot be used to fill a $\star$ hole in a multi-hole
  context. They are tested with the same rule as in $\lambdabla$
  \vspace{1mm}
  \begin{mathpar}
    \inferrule{\env_i = \reifCtx E}
    {(\enve, \env) \ltsctx i {\cval\mhc} (\enve, \env, \inctx {E} {\inctxe
        {\cval\mhc}{\enve, \env}})}
  \end{mathpar}
  \vspace{1mm}

  \noindent
  except it would be an active transition in $\lamshift$, as testing with a
  value corresponds to the throw reduction rule. So unlike in $\lambdabla$, we
  have two transitions to test contexts, in this version of $\lamshift$: one,
  active, to test a pure context in $\env$, which is used for the bisimulation,
  and one, passive, to test any evaluation context in $\enve$, which is useful
  only for up-to techniques.

\end{rem}

The definitions of the up to techniques are as expected, with weakening and
strengthening for contexts as well as for values. We write $\bbisimO$ and
$\bbisimR$ for the $\star$-bisimilarities obtained from the transitions for
respectively the original and relaxed semantics. For both semantics, the
following lemma holds. 
\begin{lem} $\setF \defeq \{ \rawweak, \rawutrctxv, \rawutrctx\}$ is
  diacritically compatible, with $\strong \setF= \{ \rawweak\}$. 
\end{lem}
\noindent As before, this lemma implies that $\bbisimO$ and $\bbisimR$ are sound
w.r.t. respectively $\ctxEqO$ and $\ctxEqR$, and completeness proofs are as
usual.
\begin{thm}
  $e_1 \ctxEqO e_2$ iff
  $(\emptyset, \emptyset, e_1) \bbisimO (\emptyset,\emptyset, e_2)$, and
  $e_1 \ctxEqR e_2$ iff
  $(\emptyset, \emptyset, e_1) \bbisimR (\emptyset, \emptyset, e_2)$.
\end{thm}

\subsection{Examples} We give examples for the original semantics of
equivalences proved in~\cite{Biernacki-Lenglet:APLAS13, Biernacki-al:HAL15}, to
show that the proofs are much easier here. 

\begin{ex}
  \label{ex:skkt}
  If $k \notin \freeVars e$, then
  $(\emptyset, \emptyset, \shift k {\app k e}) \bbisimO (\emptyset, \emptyset,
  e)$.
  We show that \iftoggle{paper}{the relation }{}
  \[ \rel \mathord{\defeq} \{ (\emptyset, \emptyset, \shift k {\app k e}),
  (\emptyset, \emptyset, e) \} \cup \{ (\reset{\lam x {\reset {\inctx E x}}
    \iapp \hole}, \reset{\reset \hole}, \emptyset), (\reset E, \reset \hole,
  \emptyset) \mid x \notin \freeVars E \} \]
  is a bisimulation up to related contexts. If $e$ is not control-stuck, the
  transition
  $(\emptyset, \emptyset, \shift k {\app k e}) \ltsstuck \fmhc (\emptyset,
  \emptyset, \inctx F {\reset {\app {\lam x {\reset {\inctx E x}}} e}}))$
  is matched by the transition
  $(\emptyset, \emptyset, e) \ltsstuck \fmhc (\emptyset, \emptyset, \inctx F
  {\reset {\inctx E e}})$,
  assuming $x$ is fresh and
  $\inctx \fmhc {\emptyset, \emptyset} = \inctx F {\reset E}$ (the case
  $\inctx \fmhc {\emptyset, \emptyset} = E$ is simple). If
  $e = \inctx {E'}{\shift {k'}{e'}}$, then
  $(\emptyset, \emptyset, e) \ltsstuck \fmhc (\emptyset, \emptyset, \inctx F
  {\reset {\subst {e'}{k'}{\lam x {\reset{\inctx E{\inctx {E'} x}}}}}})$
  is matched by the sequence \iftoggle{paper}{
  $(\emptyset, \emptyset, \shift k {\app k e}) \ltsstuck \fmhc \lts\tau
  (\emptyset, \emptyset, \inctx F {\reset {\subst {e'}{k'}{\lam x {\reset{(\lam
            y {\inctx E y}) \iapp {\inctx {E'} x}}}}}})$,}
{
  $$(\emptyset, \emptyset, \shift k {\app k e}) \ltsstuck \fmhc \lts\tau
  (\emptyset, \emptyset, \inctx F {\reset {\subst {e'}{k'}{\lam x {\reset{(\lam
            y {\inctx E y}) \iapp {\inctx {E'} x}}}}}}),$$}
  with $x$, $y$ fresh and
  $\inctx \fmhc {\emptyset, \emptyset} = \inctx F {\reset E}$. In both cases,
  the resulting states are in $\utrctxv \rel$. Let
  $(\enve, \emptyset) \defeq (\reset{\lam x {\reset {\inctx E x}} \iapp \hole},
  \reset{\reset \hole}, \emptyset)$
  and $(\envef, \emptyset) \defeq (\reset E, \reset \hole, \emptyset)$. Then the
  sequence
  $(\enve, \emptyset) \ltsctxx 1 {\cval \mhc} \lts\tau (\enve, \emptyset,
  \reset{\reset{\inctx E {\inctx {\cval \mhc}{\enve, \emptyset}}}})$
  is matched by
  $(\enve, \emptyset) \ltsctxx 1 {\cval \mhc} (\enve, \emptyset, \reset{\inctx E
    {\inctx {\cval \mhc}{\enve, \emptyset}}})$,
  since the resulting states are in $\utrctxv\rel$, and we use up to related
  contexts after a $\lts\tau$ transition. Finally,
  $(\enve, \emptyset) \ltsctxx 2 {\cval \mhc} \lts\tau \lts\tau (\enve,
  \emptyset, \inctx{\cval \mhc}{\enve, \emptyset})$
  is matched by
  $(\envef, \emptyset) \ltsctxx 2 {\cval \mhc} \lts\tau (\envef, \emptyset,
  \inctx{\cval \mhc}{\envef, \emptyset})$,
  and the context splitting transitions $\ltsnotpure i$ are easy to check for
  $i \in \{ 1, 2\}$.

\end{ex}

\begin{ex}
  If $k \notin \freeVars{e_2}$, then
  $(\emptyset, \emptyset, (\lam x {\shift k {e_1}}) \iapp e_2) \bbisimO
  (\emptyset, \emptyset, \shift k {((\lam x {e_1}) \iapp e_2)})$.  The relation
  \begin{multline*}
    \rel \mathord{\defeq} \{ (\emptyset, \emptyset, (\lam x {\shift k {e_1}})
    \iapp e_2),
    (\emptyset, \emptyset,\shift k {((\lam x {e_1}) \iapp e_2)} ) \} \\
    \cup \{ (\reset {\inctx E {(\lam x {\shift k {e_1}}) \iapp \hole}},
    \emptyset), (\reset{\lam x {\subst {e_1} k {\lam y {\reset{\inctx E y}}}} \iapp
    \hole}, \emptyset) \mid y \notin \freeVars E \}
  \end{multline*}
  is a bisimulation up to related contexts. As in the previous example, a case
  analysis on whether $e_1$ is control-stuck or not shows that the
  $\ltsstuck \fmhc$ transitions from
  $(\emptyset, \emptyset, (\lam x {\shift k {e_1}}) \iapp e_2)$ and
  $(\emptyset, \emptyset, \shift k {((\lam x {e_1}) \iapp e_2)})$ produce states
  in $\utrctxv \rel$. If
  $(\enve, \emptyset) \defeq (\reset {\inctx E {(\lam x {\shift k {e_1}}) \iapp
      \hole}}, \emptyset)$ and $(\envef, \emptyset) \defeq (\reset{\lam
  x {\subst {e_1} k {\lam y {\reset{\inctx E y}}}} \iapp \hole}, \emptyset)$,
  then  
  \vspace{1mm}
  \begin{align*}
    (\enve,
    \emptyset) & \ltsctxx 1 {\cval \mhc} \lts\tau \lts\tau (\enve, \emptyset,
    \reset{\subst {\subst {e_1} x {\inctx {\cval \mhc}{\enve, \emptyset}}} k {\lam
    y {\reset {\inctx E y}}}}) \\
    (\envef,
    \emptyset) & \ltsctxx 1 {\cval \mhc} \lts\tau (\envef, \emptyset,
    \reset{\subst {\subst {e_1} x {\inctx {\cval \mhc}{\envef, \emptyset}}} k {\lam
    y {\reset {\inctx E y}}}}) \\
  \end{align*}
  The resulting states are in $\utrctxv\rel$, as wished. A completely written
  proof of this result takes less than a page, while the proof of the same
  result in~\cite{Biernacki-al:HAL15} requires several pages, because of the
  lack of useful up-to techniques. 
\end{ex}

\section{Related Work and Conclusion}
\label{s:conclusion}

\subsubsection*{Related work.} We discuss our previous work on \textshift and
\textreset at the beginning of
Section~\ref{s:shift-reset}. In~\cite{Yachi-Sumii:APLAS16}, the authors propose
an environmental bisimilarity for a calculus with $\mathsf{call/cc}$, an
operator which captures the whole surrounding context. The difficulty in such a
language is that reduction is not preserved by evaluation context:
$e \rawred e'$ does not imply $\inctx E e \rawred \inctx E {e'}$, as $E$ may be
captured by $e$. As a result, the environmental bisimilarity
of~\cite{Yachi-Sumii:APLAS16} factors in these evaluation contexts when testing
values. This relation is also not coinductive, making it closer to contextual
equivalence than to a regular environmental bisimilarity. An accompanying
bisimulation up to context is also defined, but it is barely used in the
examples of~\cite{Yachi-Sumii:APLAS16}. The equivalence proofs of these examples
are thus almost as difficult as with contextual equivalence. It is not clear if
and how $\star$-bisimilarity can improve on these results; we plan to
investigate further this question.

Environmental bisimilarity has been defined in several calculi with dynamic
resource generation, like stores and
references~\cite{Koutavas-Wand:POPL06,Koutavas-Wand:ESOP06,Sumii:CSL09},
information hiding constructs~\cite{Sumii-Pierce:TCS07,Sumii-Pierce:JACM07}, or
name creation~\cite{Benton-Koutavas:MSR2008,Pierard-Sumii:LICS12}. In these
works, an expression is paired with its generated resources, and behavioral
equivalences are defined on these pairs. Our approach is different since we do
not carry sets of generated prompts when manipulating expressions (e.g., in the
semantic rules of Section~\ref{s:calculus}); instead, we rely on side-conditions
and permutations to avoid collisions between prompts. This is possible because
all we need to know is if a prompt is known to an outside observer or not, and
the correspondences between the public prompts of two related expressions; this
can be done through the environment of the bisimilarity. This approach cannot be
adapted to more complex generated resources, which are represented by a mapping
(e.g., for stores or existential types), but we believe it can be used for name
creation in $\pi$-calculus~\cite{Pierard-Sumii:LICS12}.

A line of work on program equivalence for which relating evaluation contexts is
crucial, as in our work, are logical relations based on the notion of
biorthogonality~\cite{Pitts-Stark:HOOTS98}. In particular, this concept has been
successfully used to develop techniques for establishing program equivalence in
ML-like languages with $\mathsf {call/cc}$~\cite{Dreyer-al:JFP12}, and for proving the
coherence of control-effect subtyping~\cite{Biernacki-Polesiuk:TLCA15}. Hur et
al. combine logical relations and behavioral equivalences in the definition of
\emph{parametric bisimulation}~\cite{Hur-al:POPL12}, where terms are reduced to
normal forms that are then decomposed into subterms related by logical
relations. This framework has been extended to abortive control
in~\cite{Hur-al:TR14}, where \emph{stuttering} is used to allow terms not to
reduce for a finite amount of time when comparing them in a bisimulation
proof. This is reminiscent of our distinction between active and passive
transitions, as passive transitions can be seen as ``not reducing'', but there
is still some testing involved in these transitions. Besides, the concern is
different, since the active/passive distinction prevents the use of up-to
techniques, while stuttering has been proposed to improve plain parametric
bisimulations.

\subsubsection*{Conclusion and future work}  We have developed a behavioral
theory for Dybvig et al.'s calculus of multi-prompted delimited
control, where the enabling technology for proving program equivalence
are environmental bisimulations, presented in Madiot's style. The
obtained results generalize our previous work in that they account for
multiple prompts and local visibility of dynamically generated
prompts. Moreover, the results of Section~\ref{s:better} considerably
enhance reasoning about captured contexts by treating them as
first-class objects at the level of bisimulation proofs (thanks to the
construct $\holecont_i$) and not only at the level of terms. The
resulting notion of bisimulation up to related contexts improves on
the existing bisimulation up to context in the presence of control
operators, as we can see when comparing Example~\ref{ex:beta-omega} to
the proof of the same result
in~\cite{Biernacki-Lenglet:APLAS13,Biernacki-al:HAL15}. Moreover, as
demonstrated in Section~\ref{s:shift-reset}, the approach of
Section~\ref{s:better} smoothly carries over to more traditional
calculi with delimited-control operators, where, in contrast to
$\lambdabla$, captured continuations are represented as functions.

We would like to see if this work scales to other formulations of control and
continuations, such as symmetric calculi~\cite{Filinski:CTCS89,
  Curien-Herbelin:ICFP00, Wadler:ICFP03}.
We believe bisimulation up to related contexts could be useful also for
constructs akin to control operators, like passivation in
$\pi$-calculus~\cite{Pierard-Sumii:LICS12}. The soundness of this
up-to technique has been proved in an extension of Madiot's framework;
we plan to investigate further this extension, to see how useful it
could be in defining up-to techniques for other languages. Finally, it
may be possible to apply the tools developed in this paper
to~\cite{Kobori-al:WoC15}, where a single-prompted calculus is
translated into a multi-prompted one, but no operational
correspondence is given to guarantee the soundness of the translation.

\subsubsection*{Acknowledgments} We would like to thank Jean-Marie Madiot for
the insightful discussions about his work, and Ma{\l}gorzata
Biernacka, Klara Zieli{\'n}ska, and the anonymous reviewers of FSCD and LMCS
for the helpful comments on the presentation of this work.


\bibliographystyle{abbrv} \bibliography{bibrefs}

\begin{thebibliography}{10}

\bibitem{Aristizabal-al:FSCD16}
A.~Aristiz\'abal, D.~Biernacki, S.~Lenglet, and P.~Polesiuk.
\newblock Environmental bisimulations for delimited-control operators with
  dynamic prompt generation.
\newblock In D.~Kesner and B.~Pientka, editors, {\em 1st International
  Conference on Formal Structures for Computation and Deduction (FSCD 2016)},
  volume~52 of {\em Leibniz International Proceedings in Informatics (LIPIcs)},
  pages 9:1--9:17, Porto, Portugal, 2016. Schloss Dagstuhl--Leibniz-Zentrum
  fuer Informatik.

\bibitem{Aristizabal-al:Inria16}
A.~Aristiz\'abal, D.~Biernacki, S.~Lenglet, and P.~Polesiuk.
\newblock Environmental bisimulations for delimited-control operators with
  dynamic prompt generation.
\newblock Research Report RR-8905, Inria, Nancy, France, Apr. 2016.
\newblock Available at \url{http://hal.inria.fr/hal-01305137}.

\bibitem{Asai-Kameyama:APLAS07}
K.~Asai and Y.~Kameyama.
\newblock Polymorphic delimited continuations.
\newblock In Z.~Shao, editor, {\em Proceedings of the 5th Asian Symposium on
  Programming Languages and Systems (APLAS'07)}, volume 4807 of {\em Lecture
  Notes in Computer Science}, pages 239--254, Singapore, Dec. 2007.

\bibitem{Balat-al:POPL04}
V.~Balat, R.~D. Cosmo, and M.~P. Fiore.
\newblock Extensional normalisation and type-directed partial evaluation for
  typed lambda calculus with sums.
\newblock In X.~Leroy, editor, {\em Proceedings of the Thirty-First Annual ACM
  Symposium on Principles of Programming Languages}, pages 64--76, Venice,
  Italy, Jan. 2004. ACM Press.

\bibitem{Benton-Koutavas:MSR2008}
N.~Benton and V.~Koutavas.
\newblock A mechanized bisimulation for the nu-calculus.
\newblock Technical Report MSR-TR-2008-129, Microsoft Research, Sept. 2008.

\bibitem{Biernacka-al:LMCS05}
M.~Biernacka, D.~Biernacki, and O.~Danvy.
\newblock An operational foundation for delimited continuations in the {CPS}
  hierarchy.
\newblock {\em Logical Methods in Computer Science}, 1(2:5):1--39, Nov. 2005.

\bibitem{Biernacki-Danvy:JFP06}
D.~Biernacki and O.~Danvy.
\newblock A simple proof of a folklore theorem about delimited control.
\newblock {\em Journal of Functional Programming}, 16(3):269--280, 2006.

\bibitem{Biernacki-Lenglet:FOSSACS12}
D.~Biernacki and S.~Lenglet.
\newblock Applicative bisimulations for delimited-control operators.
\newblock In L.~Birkedal, editor, {\em Foundations of Software Science and
  Computation Structures, 15th International Conference (FOSSACS'12)}, volume
  7213 of {\em Lecture Notes in Computer Science}, pages 119--134, Tallinn,
  Estonia, Mar. 2012. Springer.

\bibitem{Biernacki-Lenglet:FLOPS12}
D.~Biernacki and S.~Lenglet.
\newblock Normal form bisimulations for delimited-control operators.
\newblock In T.~Schrijvers and P.~Thiemann, editors, {\em Functional and Logic
  Programming, 13th International Symposium (FLOPS'12)}, volume 7294 of {\em
  Lecture Notes in Computer Science}, pages 47--61, Kobe, Japan, May 2012.
  Springer.

\bibitem{Biernacki-Lenglet:APLAS13}
D.~Biernacki and S.~Lenglet.
\newblock Environmental bisimulations for delimited-control operators.
\newblock In C.~Shan, editor, {\em Proceedings of the 11th Asian Symposium on
  Programming Languages and Systems (APLAS'13)}, volume 8301 of {\em Lecture
  Notes in Computer Science}, pages 333--348, Melbourne, VIC, Australia, Dec.
  2013. Springer.

\bibitem{Biernacki-al:HAL15}
D.~Biernacki, S.~Lenglet, and P.~Polesiuk.
\newblock Bisimulations for delimited-control operators.
\newblock Avalaible at \url{https://hal.inria.fr/hal-01207112}, 2015.

\bibitem{Biernacki-Polesiuk:TLCA15}
D.~Biernacki and P.~Polesiuk.
\newblock Logical relations for coherence of effect subtyping.
\newblock In T.~Altenkirch, editor, {\em Typed Lambda Calculi and Applications,
  13th International Conference, TLCA 2015}, volume~38 of {\em Leibniz
  International Proceedings in Informatics}, pages 107--122, Warsaw, Poland,
  July 2015. Schloss Dagstuhl -- Leibniz-Zentrum f{\"u}r Informatik.

\bibitem{Curien-Herbelin:ICFP00}
P.-L. Curien and H.~Herbelin.
\newblock The duality of computation.
\newblock In P.~Wadler, editor, {\em Proceedings of the 2000 ACM SIGPLAN
  International Conference on Functional Programming (ICFP'00)}, pages
  233--243, Montr\'eal, Canada, Sept. 2000. ACM Press.

\bibitem{Danvy-Filinski:LFP90}
O.~Danvy and A.~Filinski.
\newblock Abstracting control.
\newblock In M.~Wand, editor, {\em Proceedings of the 1990 ACM Conference on
  Lisp and Functional Programming}, pages 151--160, Nice, France, June 1990.
  ACM Press.

\bibitem{Dreyer-al:JFP12}
D.~Dreyer, G.~Neis, and L.~Birkedal.
\newblock The impact of higher-order state and control effects on local
  relational reasoning.
\newblock {\em Journal of Functional Programming}, 22(4-5):477--528, 2012.

\bibitem{Dybvig-al:JFP06}
R.~K. Dybvig, S.~Peyton-Jones, and A.~Sabry.
\newblock A monadic framework for delimited continuations.
\newblock {\em Journal of Functional Programming}, 17(6):687--730, 2007.

\bibitem{Felleisen:POPL88}
M.~Felleisen.
\newblock The theory and practice of first-class prompts.
\newblock In J.~Ferrante and P.~Mager, editors, {\em Proceedings of the
  Fifteenth Annual ACM Symposium on Principles of Programming Languages}, pages
  180--190, San Diego, California, Jan. 1988. ACM Press.

\bibitem{Filinski:CTCS89}
A.~Filinski.
\newblock Declarative continuations: An investigation of duality in programming
  language semantics.
\newblock In D.~H. Pitt et~al., editors, {\em Category Theory and Computer
  Science}, number 389 in Lecture Notes in Computer Science, pages 224--249,
  Manchester, UK, Sept. 1989. Springer-Verlag.

\bibitem{Filinski:POPL94}
A.~Filinski.
\newblock Representing monads.
\newblock In H.-J. Boehm, editor, {\em Proceedings of the Twenty-First Annual
  ACM Symposium on Principles of Programming Languages}, pages 446--457,
  Portland, Oregon, Jan. 1994. ACM Press.

\bibitem{Flatt-al:ICFP07}
M.~Flatt, G.~Yu, R.~B. Findler, and M.~Felleisen.
\newblock Adding delimited and composable control to a production programming
  environment.
\newblock In N.~Ramsey, editor, {\em Proceedings of the 2007 ACM SIGPLAN
  International Conference on Functional Programming (ICFP'07)}, pages
  165--176, Freiburg, Germany, Sept. 2007. ACM Press.

\bibitem{Gunter-al:FPCA95}
C.~Gunter, D.~R\'emy, and J.~G. Riecke.
\newblock A generalization of exceptions and control in {ML}-like languages.
\newblock In S.~{Peyton Jones}, editor, {\em Proceedings of the Seventh ACM
  Conference on Functional Programming and Computer Architecture}, pages
  12--23, La Jolla, California, June 1995. ACM Press.

\bibitem{Hur-al:POPL12}
C.~Hur, D.~Dreyer, G.~Neis, and V.~Vafeiadis.
\newblock The marriage of bisimulations and {K}ripke logical relations.
\newblock In J.~Field and M.~Hicks, editors, {\em Proceedings of the
  Thirty-Ninth Annual ACM Symposium on Principles of Programming Languages},
  pages 59--72, Philadelphia, PA, USA, Jan. 2012. ACM Press.

\bibitem{Hur-al:TR14}
C.-K. Hur, G.~Neis, D.~Dreyer, and V.~Vafeiadis.
\newblock A logical step forward in parametric bisimulations.
\newblock Technical Report MPI-SWS-2014-003, Max Planck Institute for Software
  Systems (MPI-SWS), Saarbrücken, Germany, Jan. 2014.

\bibitem{Kameyama:HOSC07}
Y.~Kameyama.
\newblock Axioms for control operators in the {CPS} hierarchy.
\newblock {\em Higher-Order and Symbolic Computation}, 20(4):339--369, 2007.

\bibitem{Kameyama-Hasegawa:ICFP03}
Y.~Kameyama and M.~Hasegawa.
\newblock A sound and complete axiomatization of delimited continuations.
\newblock In Shivers \cite{ICFP:03}, pages 177--188.

\bibitem{Kiselyov:FLOPS10}
O.~Kiselyov.
\newblock Delimited control in {OC}aml, abstractly and concretely: System
  description.
\newblock In M.~Blume and G.~Vidal, editors, {\em Functional and Logic
  Programming, 10th International Symposium, FLOPS 2010}, volume 6009 of {\em
  Lecture Notes in Computer Science}, pages 304--320, Sendai, Japan, Apr. 2010.
  Springer.

\bibitem{Kobori-al:WoC15}
I.~Kobori, Y.~Kameyama, and O.~Kiselyov.
\newblock {ATM} without tears: prompt-passing style transformation for typed
  delimited-control operators.
\newblock In O.~Danvy, editor, {\em 2015 Workshop on Continuations:
  Pre-proceedings}, London, UK, Apr. 2015.

\bibitem{Koutavas-al:MFPS11}
V.~Koutavas, P.~B. Levy, and E.~Sumii.
\newblock From applicative to environmental bisimulation.
\newblock In M.~Mislove and J.~Ouaknine, editors, {\em Proceedings of the 27th
  Annual Conference on Mathematical Foundations of Programming Semantics (MFPS
  XXVII)}, volume 276 of {\em Electronic Notes in Theoretical Computer
  Science}, pages 215--235, Pittsburgh, PA, USA, May 2011.

\bibitem{Koutavas-Wand:ESOP06}
V.~Koutavas and M.~Wand.
\newblock Bisimulations for untyped imperative objects.
\newblock In P.~Sestoft, editor, {\em 15th European Symposium on Programming,
  ESOP 2006}, volume 3924 of {\em Lecture Notes in Computer Science}, pages
  146--161, Vienna, Austria, Mar. 2006. Springer.

\bibitem{Koutavas-Wand:POPL06}
V.~Koutavas and M.~Wand.
\newblock Small bisimulations for reasoning about higher-order imperative
  programs.
\newblock In J.~G. Morrisett and S.~L.~P. Jones, editors, {\em Proceedings of
  the 33rd Annual ACM Symposium on Principles of Programming Languages}, pages
  141--152, Charleston, SC, USA, Jan. 2006. ACM Press.

\bibitem{Madiot-al:CONCUR14}
J.~Madiot, D.~Pous, and D.~Sangiorgi.
\newblock Bisimulations up-to: Beyond first-order transition systems.
\newblock In P.~Baldan and D.~Gorla, editors, {\em 25th International
  Conference on Concurrency Theory}, volume 8704 of {\em Lecture Notes in
  Computer Science}, pages 93--108, Rome, Italy, Sept. 2014. Springer.

\bibitem{Madiot:PhD}
J.-M. Madiot.
\newblock {\em Higher-Order Languages: Dualities and Bisimulation
  Enhancements}.
\newblock PhD thesis, Universit{\'e} de Lyon and Universit{\`a} di Bologna,
  2015.

\bibitem{Pierard-Sumii:LICS12}
A.~Pi{\'e}rard and E.~Sumii.
\newblock A higher-order distributed calculus with name creation.
\newblock In {\em Proceedings of the 27th {IEEE} Symposium on Logic in Computer
  Science (LICS 2012)}, pages 531--540, Dubrovnik, Croatia, June 2012. IEEE
  Computer Society Press.

\bibitem{Pitts-Stark:HOOTS98}
A.~Pitts and I.~Stark.
\newblock Operational reasoning for functions with local state.
\newblock In A.~Gordon and A.~Pitts, editors, {\em Higher Order Operational
  Techniques in Semantics}, pages 227--273. Publications of the Newton
  Institute, Cambridge University Press, 1998.

\bibitem{Sangiorgi-Pous:11}
D.~Pous and D.~Sangiorgi.
\newblock Enhancements of the bisimulation proof method.
\newblock In D.~Sangiorgi and J.~Rutten, editors, {\em Advanced Topics in
  Bisimulation and Coinduction}, chapter~6, pages 233--289. Cambridge
  University Press, 2011.

\bibitem{Sangiorgi-al:TOPLAS11}
D.~Sangiorgi, N.~Kobayashi, and E.~Sumii.
\newblock Environmental bisimulations for higher-order languages.
\newblock {\em ACM Transactions on Programming Languages and Systems},
  33(1):1--69, Jan. 2011.

\bibitem{ICFP:03}
O.~Shivers, editor.
\newblock {\em Proceedings of the 2003 ACM SIGPLAN International Conference on
  Functional Programming (ICFP'03)}, Uppsala, Sweden, Aug. 2003. ACM Press.

\bibitem{Sumii:CSL09}
E.~Sumii.
\newblock A complete characterization of observational equivalence in
  polymorphic lambda-calculus with general references.
\newblock In E.~Gr{\"{a}}del and R.~Kahle, editors, {\em Computer Science
  Logic, {CSL'09}}, volume 5771 of {\em Lecture Notes in Computer Science},
  pages 455--469, Coimbra, Portugal, Sept. 2009. Springer.

\bibitem{Sumii-Pierce:TCS07}
E.~Sumii and B.~C. Pierce.
\newblock A bisimulation for dynamic sealing.
\newblock {\em Theoretical Computer Science}, 375(1-3):169--192, 2007.

\bibitem{Sumii-Pierce:JACM07}
E.~Sumii and B.~C. Pierce.
\newblock A bisimulation for type abstraction and recursion.
\newblock {\em Journal of the ACM}, 54(5), 2007.

\bibitem{Wadler:ICFP03}
P.~Wadler.
\newblock Call-by-value is dual to call-by-name.
\newblock In Shivers \cite{ICFP:03}, pages 189--201.

\bibitem{Yachi-Sumii:APLAS16}
T.~Yachi and E.~Sumii.
\newblock A sound and complete bisimulation for contextual equivalence in
  $\lambda$-calculus with call/cc.
\newblock In A.~Igarashi, editor, {\em Proceedings of the 14th Asian Symposium
  on Programming Languages and Systems (APLAS'16)}, volume 10017 of {\em
  Lecture Notes in Computer Science}, pages 171--186, Hanoi, Vietnam, Nov.
  2016.

\end{thebibliography}

\end{document}